\documentclass{article}

\usepackage[a4paper]{geometry}
\usepackage[english]{babel}
\usepackage{hyperref}

\usepackage[nobreak]{cite} %Turns a citation that looks like [3,4,5] into [3-5]. The nobreak option prevents cite from breaking at the end of a line.

\usepackage{microtype}

\usepackage{amssymb,amsthm,amsmath,bm}

\usepackage{mathtools} % for \smashoperator macro

\usepackage{tikz}
\usetikzlibrary{arrows,decorations,calc}
\usetikzlibrary{positioning}

\definecolor{darkblue}{rgb}{0,0,.8}
\definecolor{red}{rgb}{1,0,0}
\renewcommand{\i}{\mathrm{i}}
\newcommand{\diff}{\mathrm{d}}
\newcommand{\ee}{\mathrm{e}}
\newcommand{\piJ}{P}

\hypersetup{
colorlinks=true,
citecolor=red,
linkcolor=darkblue,
urlcolor=black
}

\usepackage[noabbrev,capitalise]{cleveref}

\newtheorem{theorem}{Theorem}
\newtheorem{proposition}{Proposition}[section]
\newtheorem{lemma}[proposition]{Lemma}

\newtheorem{conjecture}[proposition]{Conjecture}

\title{\Large{\textbf{
The open XXZ chain at $\boldsymbol{\Delta=-1/2}$ and the boundary quantum Knizhnik-Zamolodchikov equations}}}
\author{\normalsize \textsc{Christian Hagendorf} and \textsc{Jean Li\'enardy} \medskip \\
{\normalsize
  \begin{minipage}{.9\textwidth} 
  \begin{center}
  \textit{
   Universit\'e catholique de Louvain\\
  Institut de Recherche en Math\'ematique et Physique\\
  Chemin du Cyclotron 2, 1348 Louvain-la-Neuve, Belgium} \\
  \bigskip
  \href{mailto:christian.hagendorf@uclouvain.be}{\normalsize 
\texttt{christian.hagendorf@uclouvain.be}}, \href{mailto:jean.lienardy@uclouvain.be}{\normalsize 
\texttt{jean.lienardy@uclouvain.be}}
  \end{center}
  \end{minipage}
}
}
\date{}                                           % Activate to display a given date or no date

\begin{document}

\maketitle

\begin{abstract}
The open XXZ spin chain with the anisotropy parameter $\Delta=-\frac12$ and diagonal boundary magnetic fields that depend on a parameter $x$ is studied. For real $x>0$, the exact finite-size ground-state eigenvalue of the spin-chain Hamiltonian is explicitly computed. In a suitable normalisation, the ground-state components are characterised as polynomials in $x$ with integer coefficients. Linear sum rules and special components of this eigenvector are explicitly computed in terms of determinant formulas. These results follow from the construction of a contour-integral solution to the boundary quantum Knizhnik-Zamolodchikov equations associated with the $R$-matrix and diagonal $K$-matrices of the six-vertex model. A relation between this solution and a weighted enumeration of totally-symmetric alternating sign matrices is conjectured.
\end{abstract}

%%%%%%%%%%%%%%%%%%%%%%%%%%%%%%%%%%%%%%%%%%%%%%%%%%%%%%%%%%%%%%%%%%%%%%%%%%%%%%%%%%%%%%%%%%%%%
\section{Introduction}
\label{sec:intro}
%%%%%%%%%%%%%%%%%%%%%%%%%%%%%%%%%%%%%%%%%%%%%%%%%%%%%%%%%%%%%%%%%%%%%%%%%%%%%%%%%%%%%%%%%%%%%

The XXZ spin chain is a one-dimensional system of quantum spins $\frac12$ with nearest-neighbour interactions. In this article, we consider this spin chain with open boundary conditions and diagonal boundary magnetic fields. For $N\geqslant 1$ sites, its Hamiltonian is given by\footnote{For $N=1$ sites, the bulk interaction term is absent by convention. The Hamiltonian reduces to $H=(p+\bar p)\sigma^z$.}
\begin{equation}
  \label{eqn:XXZ}
  H = - \frac12\sum_{i=1}^{N-1}\left(\sigma_i^x\sigma_{i+1}^x+\sigma_i^y\sigma_{i+1}^y+\Delta \sigma_i^z\sigma_{i+1}^z\right) + p \sigma_1^z+\bar p\sigma_N^z.
\end{equation}
Here, $\sigma_i^x,\,\sigma_i^y,\sigma^z_i$ are the standard Pauli matrices, acting on the site $i=1,\dots,N$. Moreover, $\Delta$ is the spin chain's anisotropy parameter. The parameters $p,\,\bar p$ measure the strengths of its boundary magnetic fields. We focus on the case
\begin{equation}
  \label{eqn:CombinatorialPoint}
  \Delta = -\frac12,\quad p = \frac12\left(\frac12-x\right),\quad \bar p = \frac12\left(\frac12-\frac1x\right),
\end{equation}
where $x$ is an arbitrary parameter. This case is of interest for several reasons. First, the spin chain's anisotropy parameter is adjusted to the so-called \textit{combinatorial point} $\Delta=-\frac12$, which has received considerable attention \cite{zinnjustin:10}. For periodic and twisted boundary conditions, its ground state possesses remarkable relations to the enumerative combinatorics of alternating sign matrices and plane partitions \cite{stroganov:01,razumov:00,razumov:01,batchelor:01,degier:02,difrancesco:06,cantini:12_1,morin:20}. For open boundary conditions, similar relations between the spin chain's ground states and combinatorics are expected for the parameters \eqref{eqn:CombinatorialPoint}, provided that $x>0$ \cite{degier:04,nichols:05}. The purpose of this article is to unveil some of these relations through an explicit construction of the spin chain's ground-state vector. Second, the continuum limit of the spin chain with $\Delta=-\frac12$ is expected to be described by a superconformal field theory, irrespectively of the values taken by $p,\,\bar p$, provided that they are real \cite{degier:04,degier:05}. Surprisingly, for $x=1$ the Hamiltonian is supersymmetric even for chains of finite length \cite{yang:04,weston:17,hagendorf:17}. The supersymmetry allows one to analyse the spin chain's ground states through (co)homology methods. In \cite{hagendorf:17}, we used these methods to establish a set of sum rules for the ground states that allow quantifying their quantum entanglement properties. A special case of these sum rules allowed us to compute the spin chain's bipartite fidelity \cite{dubail:11,dubail:13}. In this article, we observe a generalisation of these sum rules to arbitrary $x$. Third, for $x = \ee^{\i\pi/3}$, we obtain a special case of the quantum-group invariant Pasquier-Saleur Hamiltonian. It has received considerable attention because of its connection to loop models, the Temperley-Lieb algebra and conformal field theory \cite{pasquier:90}.

Our construction of the spin-chain ground state extends a strategy of Razumov, Stroganov and Zinn-Justin \cite{razumov:07} to the case of open boundary conditions. We consider the so-called \textit{boundary quantum Knizhnik-Zamolodchikov equations} for the $R$-matrix and diagonal $K$-matrices of the six-vertex model \cite{cherednik:92,jimbo:95,stokman:15,reshetikhin:15,reshetikhin:18}. These equations depend on a deformation parameter $q$ and a boundary parameter $\beta$. We find a solution to these equations in terms of multiple contour integrals. For $q=\ee^{\pm 2\pi\i/3}$, this solution is an eigenvector of the transfer matrix of an inhomogeneous six-vertex model on a strip. In the so-called homogeneous limit, it becomes an eigenvector of the spin-chain Hamiltonian \eqref{eqn:XXZ} with the parameters \eqref{eqn:CombinatorialPoint}, where $x$ is related to $\beta$. The explicit contour-integral representations for its components are a powerful tool to investigate the vector's properties.

This article is organised as follows. In \cref{sec:SpinChainGS}, we present our main results, as well as certain observations, on the ground state of the XXZ Hamiltonian \eqref{eqn:XXZ} with the parameters \eqref{eqn:CombinatorialPoint}. The purpose of \cref{sec:bqKZ} is to construct and analyse a solution to the boundary quantum Knizhnik-Zamolodchikov equations in terms of multiple contour integrals. We show in \cref{sec:6V} that for $q=\ee^{\pm 2\pi\i/3}$ this solution yields an eigenvector of the transfer matrix of an inhomogeneous  six-vertex model on a strip. In \cref{sec:HomLimit}, we compute the homogeneous limit of the solution for generic $q$. Its specialisation to $q=\ee^{\pm 2\pi\i/3}$ allows us to prove the statements for the ground state of \cref{sec:SpinChainGS}. Moreover, we conjecture a relation between the homogeneous limit and the enumeration of totally-symmetric alternating sign matrices. We present our conclusions in \cref{sec:Conclusion}.

%%%%%%%%%%%%%%%%%%%%%%%%%%%%%%%%%%%%%%%%%%%%%%%%%%%%%%%%%%%%%%%%%%%%%%%%%%%%%%%%%%%%%%%%%%%%%
\section{The spin-chain ground state}
\label{sec:SpinChainGS}
%%%%%%%%%%%%%%%%%%%%%%%%%%%%%%%%%%%%%%%%%%%%%%%%%%%%%%%%%%%%%%%%%%%%%%%%%%%%%%%%%%%%%%%%%%%%%

The purpose of this section is to present and discuss our main results and observations about the ground state of the Hamiltonian \eqref{eqn:XXZ}. We fix our notations and conventions in \cref{sec:Notations}. In \cref{sec:MainResults}, we state our results in four theorems and one conjecture. We relegate the theorems' proofs to \cref{sec:HomLimit}.

%%%%%%%%%%%%%%%%%%%%%%%%%%%%%%%%%%%%%%%%%%%%%%%%%%%%%%%%%%%%%%%%%%%%%%%%%%%%%%%%%%%%%%%%%%%%%
\subsection{Notations}
\label{sec:Notations}
%%%%%%%%%%%%%%%%%%%%%%%%%%%%%%%%%%%%%%%%%%%%%%%%%%%%%%%%%%%%%%%%%%%%%%%%%%%%%%%%%%%%%%%%%%%%%

The Hamiltonian \eqref{eqn:XXZ} is an operator acting on the space $V^N = V_1\otimes V_2 \otimes \cdots \otimes V_N$ where each $V_i = V =\mathbb C^2$ is a copy of the space of states of a quantum spin $1/2$. Its  canonical basis vectors are
\begin{equation}
|{\uparrow}\rangle =
  \begin{pmatrix}
    1 \\ 0
  \end{pmatrix},
  \quad
  |{\downarrow}\rangle =
  \begin{pmatrix}
    0 \\ 1
  \end{pmatrix}.
\end{equation}
The canonical basis vectors of $V^N$ are given by tensor products  $|s_1\cdots s_N\rangle = |s_1\rangle \otimes \cdots \otimes |s_N\rangle$, where $s_i \in \{\uparrow,\downarrow\}$ for each $i=1,\dots,N$. We recall that the components of a vector $|\psi\rangle \in V^N$ are the coefficients of its expansion along these canonical basis vectors: 
\begin{equation}
  |\psi\rangle = \sum_{ s_1,\dots, s_N \in \{\uparrow,\downarrow\} } \psi_{s_1\cdots s_N} | s_1\cdots s_N \rangle.
\end{equation}
For any vector $|\phi\rangle \in V^N$, we define a co-vector $\langle\phi| =|\phi\rangle^t \in (V^N)^\ast$ by transposition (without complex conjugation). The dual pairing between a co-vector $\langle \phi|$ and a vector $|\psi\rangle$ is defined by
\begin{equation}
 \langle \phi|\psi\rangle = \sum_{ s_1,\dots, s_N\in\{\uparrow,\downarrow\}} \phi_{ s_1\cdots  s_N}\psi_{s_1\cdots  s_N}.
\end{equation}
We refer to this pairing as the scalar product of the vectors $|\phi\rangle$ and $|\psi\rangle$% 
, even though it only defines a scalar product on a real subspace of $V^N$. We write $\|\psi\|^2 = \langle \psi|\psi\rangle$ for the square norm of a vector in this subspace.

We recall that the Pauli matrices are given by
\begin{equation}
  \sigma^x =
  \begin{pmatrix}
    0 & 1\\
    1 & 0
  \end{pmatrix},
  \quad
  \sigma^y =
  \begin{pmatrix}
    0 & -\i\\
    \i & 0
  \end{pmatrix},
  \quad
  \sigma^z =
  \begin{pmatrix}
    1 & 0\\
    0 & -1
  \end{pmatrix}.
\end{equation}
For $\kappa=x,y,z$ and $i=1,\dots,N$, we define
\begin{equation}
  \sigma^\kappa_i = \underset{i-1}{\underbrace{\bm 1 \otimes \cdots \otimes \bm 1}} \otimes \sigma^\kappa \otimes  \underset{N-i}{\underbrace{\bm 1 \otimes \cdots \otimes \bm 1}},
\end{equation}
where $\bm 1$ is the identity operator on $V$. We use the Pauli matrices to define the magnetisation operator $\mathcal M$ and the spin-reversal operator $\mathcal R$ on $V^N$:
\begin{equation}
 \label{eqn:SpinReversal}
 \mathcal M = \frac12\sum_{i=1}^N \sigma_i^z, \quad \mathcal R = \prod_{i=1}^N \sigma_i^x.
\end{equation}

Finally, let $1\leqslant M \leqslant N$ be an integer and $A$ be an operator acting on $V^M$. For each $1\leqslant i_1 < \cdots < i_M \leqslant N$, we use the standard tensor-leg notation $A_{i_1,\dots,i_M}$ for the operator that acts like $A$ on the factors $V_{i_1},\dots,V_{i_M}$ of $V^N$, and like the identity operator on the other factors.

%%%%%%%%%%%%%%%%%%%%%%%%%%%%%%%%%%%%%%%%%%%%%%%%%%%%%%%%%%%%%%%%%%%%%%%%%%%%%%%%%%%%%%%%%%%%%
\subsection{Main results}
\label{sec:MainResults}
%%%%%%%%%%%%%%%%%%%%%%%%%%%%%%%%%%%%%%%%%%%%%%%%%%%%%%%%%%%%%%%%%%%%%%%%%%%%%%%%%%%%%%%%%%%%%
We now present our main results for the ground state of the Hamiltonian \eqref{eqn:XXZ} with the parameters \eqref{eqn:CombinatorialPoint}. We find its exact ground-state eigenvalue, characterise the components of the corresponding eigenvector in a suitable normalisation, and find exact determinant formulas for certain scalar products and special components of the vector. We discuss our findings at the supersymmetric point. Finally, we conjecture an exact expression for a family of scalar products that allows one to find ground state's logarithmic bipartite fidelity.

%%%%%%%%%%%%%%%%%%%%%%%%%%%%%%%%%%%%%%%%%%%%%%%%%%%%%%%%%%%%%%%%%%%%%%%%%%%%%%%%%%%%%%%%%%%%%
\subsubsection*{The ground-state eigenvalue}
%%%%%%%%%%%%%%%%%%%%%%%%%%%%%%%%%%%%%%%%%%%%%%%%%%%%%%%%%%%%%%%%%%%%%%%%%%%%%%%%%%%%%%%%%%%%%
Throughout this article, we use the integers
\begin{equation}
  \label{eqn:Defnnbar}
  n = \lfloor N/2\rfloor, \quad \bar n = \lceil N/2\rceil.
\end{equation}

The Hamiltonian \eqref{eqn:XXZ} commutes with the magnetisation operator: $[H,\mathcal M]=0$. Therefore, $H$ and $\mathcal M$ can be simultaneously diagonalised. The eigenvalues of the magnetisation operator are $\mu = -N/2,-N/2+1,\dots,N/2$. We call the corresponding eigenspaces sectors of magnetisation $\mu$. Our first main theorem addresses an explicit expression for the ground-state eigenvalue of the Hamiltonian in the sector of magnetisation $\mu = (\bar n-n)/2$ (which corresponds to $\mu=0$ for even $N$, and $\mu = 1/2$ for odd $N$). Its existence was conjectured in \cite{nichols:05}.
\begin{theorem}[Special eigenvalue]
  \label{thm:MainTheorem1}
  The Hamiltonian \eqref{eqn:XXZ} with the parameters \eqref{eqn:CombinatorialPoint} possesses the eigenvalue
  \begin{equation}
  \label{eqn:SpecialEV}
    E_0 = -\frac{3N-1}{4}-\frac{(1-x)^2}{2x}.
  \end{equation}
  For real $x>0$, it is the non-degenerate ground-state eigenvalue in the sector of magnetisation $\mu = (\bar n-n)/2$.
\end{theorem}
We prove this theorem in \cref{sec:EV} through the explicit construction of an eigenvector and an application of the Perron-Frobenius theorem. Here, we limit ourselves to the remark that \eqref{eqn:SpecialEV} is the non-degenerate ground-state eigenvalue at the supersymmetric point $x=1$, without a restriction to a sector of magnetisation \cite{hagendorf:17}. 
Moreover, the Hamiltonian's eigenvalues depend continuously on $x$ for $x>0$. Hence, there are real numbers $0 \leqslant x_-< 1$ and $x_+>1$ such that \eqref{eqn:SpecialEV} is the non-degenerate ground-state eigenvalue for all $x_-<x<x_+$. A numerical analysis of the Hamiltonian's spectrum for small $N$ suggests that $x_-=0$ and $x_+=\infty$. Furthermore, the analysis suggests that \eqref{eqn:SpecialEV} is not the ground-state eigenvalue for $x<0$.

%%%%%%%%%%%%%%%%%%%%%%%%%%%%%%%%%%%%%%%%%%%%%%%%%%%%%%%%%%%%%%%%%%%%%%%%%%%%%%%%%%%%%%%%%%%%%
\subsubsection*{The ground-state vector}
%%%%%%%%%%%%%%%%%%%%%%%%%%%%%%%%%%%%%%%%%%%%%%%%%%%%%%%%%%%%%%%%%%%%%%%%%%%%%%%%%%%%%%%%%%%%%
According to \cref{thm:MainTheorem1}, the eigenspace of \eqref{eqn:SpecialEV} of the Hamiltonian's restriction to the sector of magnetisation $\mu=(\bar n-n)/2$ is one-dimensional. We study a basis vector $|\psi_N\rangle$ of this eigenspace whose normalisation is fixed by
\begin{equation}
(\psi_N)_{\underset{n}{\underbrace{\scriptstyle \downarrow\cdots \downarrow}}\underset{\bar n}{\underbrace{\scriptstyle \uparrow\cdots \uparrow}}} = 1.
\end{equation}
We refer to $|\psi_N\rangle$ as the ground-state vector (even though, strictly speaking, it is the ground-state vector only for real $x>0$). For small $N$, its non-zero components are easily found from the exact solution of the eigenvalue problem for $E_0$. For example, for $N=5$ sites (where $n=2,\,\bar n = 3$), we obtain
\begin{align}
  (\psi_5)_{\downarrow\downarrow\uparrow\uparrow\uparrow} &= 1,  &(\psi_5)_{\downarrow\uparrow\downarrow\uparrow\uparrow} &= 3+ x,\nonumber \\
  (\psi_5)_{\downarrow\uparrow\uparrow\downarrow\uparrow} &= 3(1 + x),  &(\psi_5)_{\downarrow\uparrow\uparrow\uparrow\downarrow} &= 3x,\nonumber \\
  (\psi_5)_{\uparrow\downarrow\downarrow\uparrow\uparrow} &= 2 + 2x+x^2, &(\psi_5)_{\uparrow\downarrow\uparrow\downarrow\uparrow} &= 3+ 5x+3x^2,
  \label{eqn:ExN5}\\
  (\psi_5)_{\uparrow\downarrow\uparrow\uparrow\downarrow} &= 3x(1+x),  &(\psi_5)_{\uparrow\uparrow\downarrow\downarrow\uparrow} &= 1 + 2x + 2x^2,\nonumber \\
  (\psi_5)_{\uparrow\uparrow\downarrow\uparrow\downarrow} &= x(1+3x) ,  &(\psi_5)_{\uparrow\uparrow\uparrow\downarrow\downarrow} &= x^2.\nonumber 
\end{align}

In \cref{sec:Components}, we express the components in terms of multiple contour integrals. Using these contour integrals, we prove the following property:
\begin{theorem}[Polynomiality]
  \label{thm:MainTheorem2}
The components $(\psi_N)_{s_1\cdots s_N}$ are polynomials in $x$ with integer coefficients. For each $0\leqslant m \leqslant n$, the degree of the polynomial $(\psi_N)_{\downarrow\cdots \downarrow s_{m+1}\cdots s_N}$ is at most $n-m$.
\end{theorem}
The solution of the eigenvalue problem for small $N$ suggests the even stronger result that the components are polynomials in $x$ with \textit{non-negative} integer coefficients.

Next, we consider the vector's behaviour under a parity transformation. We recall the that the parity operator $\mathcal P$ is the linear operator whose action on the canonical basis vectors is given by
\begin{equation}
 \mathcal P| s_1 s_2\cdots  s_N\rangle = | s_N \cdots  s_2 s_1\rangle.
\end{equation}
To express the action of the parity operator on the ground-state vector, we stress its dependence of $x$ by writing $|\psi_N\rangle = |\psi_N(x)\rangle$. In \cref{sec:Components}, we prove the following:
\begin{theorem}[Parity]
  \label{thm:MainTheorem3}
  We have $\mathcal P|\psi_N(x)\rangle = x^n|\psi_N(x^{-1})\rangle$.
\end{theorem}

%%%%%%%%%%%%%%%%%%%%%%%%%%%%%%%%%%%%%%%%%%%%%%%%%%%%%%%%%%%%%%%%%%%%%%%%%%%%%%%%%%%%%%%%%%%%%
\subsubsection*{Scalar products and special components}
%%%%%%%%%%%%%%%%%%%%%%%%%%%%%%%%%%%%%%%%%%%%%%%%%%%%%%%%%%%%%%%%%%%%%%%%%%%%%%%%%%%%%%%%%%%%%
Let us introduce the co-vector $\langle\xi(\alpha)| = \langle {\uparrow\downarrow}|+\alpha \langle{\downarrow\uparrow}|$, where $\alpha$ is a complex number. We use it to define the scalar products
\begin{align}    
  \label{eqn:DefF} 
  F_{2n} & = \bigl(\langle \xi(\alpha)|\otimes \cdots \otimes \langle \xi(\alpha)|\bigr)|\psi_{2n}\rangle, \\
 F_{2n+1} & 
 = \bigl(\langle {\uparrow}|\otimes \langle \xi(\alpha)|\otimes \cdots \otimes \langle \xi(\alpha)|\bigr)|\psi_{2n+1}\rangle.
\end{align}
These scalar products are polynomials in $\alpha$ of degree at most $n$. Their coefficients are polynomials in $x$, given by linear combinations of the ground-state components. For example, for $N=5$ we obtain
\begin{equation}
  F_5 = (\psi_5)_{\uparrow\uparrow\downarrow\uparrow\downarrow} + \left( (\psi_5)_{\uparrow\downarrow\uparrow\uparrow\downarrow}+ (\psi_5)_{\uparrow\uparrow\downarrow\downarrow\uparrow}\right)\alpha +  (\psi_5)_{\uparrow\downarrow\uparrow\downarrow\uparrow}\alpha^2.
\end{equation}
In \cref{sec:ScalarProducts}, we prove that the scalar products are given by the following closed-form expressions:
\begin{theorem} [Scalar products]
  \label{thm:MainTheorem4}
  For each $n\geqslant 0$, we have
  \begin{equation}
    F_{2n} = \det_{i,j=1}^n\left( \alpha x  \binom{i+j-2}{2i-j-2}+(\alpha +x)\binom{i+j-2}{2i-j-1}+\binom{i+j-2}{2i-j}\right),
  \end{equation}
 and 
  \begin{equation}
    F_{2n+1} = \det_{i,j=1}^n\left( \alpha x \binom{i+j-1}{2i-j-1}+(\alpha + x)\binom{i+j-1}{2i-j}+\binom{i+j-1}{2i-j+1}\right).
  \end{equation}
\end{theorem}
This theorem reveals that the scalar products are symmetric under the exchange of $x$ and $\alpha$. This comes as a surprise as $x$ and $\alpha$ enter the definition of $F_N$ in a very different, not obviously symmetric way. Furthermore, the theorem allows us to obtain determinant formulas for the components of the ground-state vector that are labelled by alternating spin configurations. Indeed, we find through the specialisations $\alpha = 0$ and $\alpha \to \infty$ the expressions
\begin{subequations}
\label{eqn:SpecialComponents}
\begin{align}
  (\psi_{2n})_{\uparrow\downarrow\cdots \uparrow\downarrow} &= \det_{i,j=1}^n\left(x\binom{i+j-2}{2i-j-1}+ \binom{i+j-2}{2i-j}\right),\\
  (\psi_{2n+1})_{\uparrow\downarrow\cdots \uparrow\downarrow\uparrow} & = \det_{i,j=1}^n\left(x\binom{i+j-1}{2i-j-1}+\binom{i+j-1}{2i-j}\right).
\end{align}
\end{subequations}
We note that, up to integer pre-factors, these determinants also appear as certain generating functions in the $O(1)$ model with reflecting boundary conditions \cite{degier:16}.

%%%%%%%%%%%%%%%%%%%%%%%%%%%%%%%%%%%%%%%%%%%%%%%%%%%%%%%%%%%%%%%%%%%%%%%%%%%%%%%%%%%%%%%%%%%%%
\subsubsection*{The supersymmetric point}
%%%%%%%%%%%%%%%%%%%%%%%%%%%%%%%%%%%%%%%%%%%%%%%%%%%%%%%%%%%%%%%%%%%%%%%%%%%%%%%%%%%%%%%%%%%%%
We now consider the supersymmetric point $x=1$. In this case, some of the properties of the ground-state vector can be expressed in terms of the number of vertically-symmetric alternating sign matrices of size $2n+1$, and the number of cyclically-symmetric transpose complement plane partitions in a $2n\times 2n\times 2n$ cube. These numbers are given by \cite{kuperberg:02,bressoudbook}
\begin{equation}
  A_{\mathrm{V}}(2n+1) = \frac{1}{2^n}\prod_{k=1}^{n}\frac{(6k-2)!(2k-1)!}{(4k-2)!(4k-1)!}, \quad N_8(2n) = \prod_{k=0}^{n-1}\frac{(3k+1)(6k)!(2k)!}{(4k)!(4k+1)!}.
\end{equation}
Our first result concerns the special components \eqref{eqn:SpecialComponents}. For $x=1$, their determinant expressions can be explicitly evaluated with the help of Krattenthaler's formula \cite{bressoudbook}. We obtain
\begin{equation}
  \label{eqn:AlternatingEntriesSUSY}
  (\psi_{2n})_{\uparrow\downarrow\cdots \uparrow\downarrow} = A_{\mathrm{V}}(2n+1), \quad (\psi_{2n+1})_{\uparrow\downarrow\cdots \uparrow\downarrow\uparrow} = N_{8}(2n+2).
\end{equation}
Similarly, the scalar product $F_N$ with $x=1$ and $\alpha = 1$ is given by
\begin{align}
 \label{eqn:ProjectionsSUSY}
\left(\langle \chi|\otimes \cdots \otimes \langle \chi|\right)|\psi_{2n}\rangle = N_8(2n+2), 
\quad 
\left(\langle{\uparrow}|\otimes \langle \chi|\otimes \cdots \otimes \langle \chi|\right)|\psi_{2n+1}\rangle
 = A_{\mathrm{V}}(2n+3),
\end{align}
where we used the shorthand notation $|\chi\rangle = |\xi(1)\rangle$. Using \eqref{eqn:AlternatingEntriesSUSY} and \eqref{eqn:ProjectionsSUSY}, we may compute the square norm of ground-state vector for $x=1$. Indeed, in \cite{hagendorf:17} we used the supersymmetry of the spin-chain Hamiltonian at this point to establish the factorisations
\begin{align}
  \|\psi_{2n}\|^2 &= (\psi_{2n})_{\uparrow\downarrow\cdots \uparrow\downarrow} \left(\langle \chi|\otimes \cdots \otimes \langle \chi|\right)|\psi_{2n}\rangle,\\
   \|\psi_{2n+1}\|^2 &= (\psi_{2n})_{\uparrow\downarrow\cdots \uparrow\downarrow\uparrow} \left(\langle{\uparrow}|\otimes\langle \chi|\otimes \cdots \otimes \langle \chi|\right)|\psi_{2n+1}\rangle.
\end{align}
They lead to the closed-form expression
\begin{equation}
  \label{eqn:SquareNormSUSY}
  \|\psi_{N}\|^2 = A_{\mathrm{V}}(2\bar n+1)N_8(2n+2).
\end{equation}
We note that the results \eqref{eqn:AlternatingEntriesSUSY} and \eqref{eqn:SquareNormSUSY} settle Conjecture 6.1 from \cite{hagendorf:17}.

%%%%%%%%%%%%%%%%%%%%%%%%%%%%%%%%%%%%%%%%%%%%%%%%%%%%%%%%%%%%%%%%%%%%%%%%%%%%%%%%%%%%%%%%%%%%%
\subsubsection*{The logarithmic bipartite fidelity}
%%%%%%%%%%%%%%%%%%%%%%%%%%%%%%%%%%%%%%%%%%%%%%%%%%%%%%%%%%%%%%%%%%%%%%%%%%%%%%%%%%%%%%%%%%%%%

The logarithmic bipartite fidelity was introduced by Dubail and St\'ephan as an entanglement measure and a tool for detecting quantum critical points of interacting quantum systems \cite{dubail:11,dubail:13} (see also \cite{parez:19, morin:20b}). 
For the open XXZ chain of the present article with $x>0$, it is given by
\begin{equation}
  \mathcal F_{N_1,N_2} = - \ln \left|\frac{\langle \psi_{N}|\left(|\psi_{N_1}\rangle \otimes |\psi_{N_2}\rangle\right) }{\|\psi_N\|\, \|\psi_{N_1}\|\, \|\psi_{N_2}\|}\right|^2,
\end{equation}
where $N_1,N_2\geqslant 1$ are integers and $N=N_1+N_2$. The leading terms of its asymptotic expansion for large $N$ have been predicted by conformal field theory.

Here, we present a conjecture that allows us to obtain a closed-form expression of the logarithmic bipartite fidelity and its multipartite generalisations for finite $N$ \cite{lienardy:20}. To this end, let $1\leqslant m \leqslant N$ and $N_1,\dots,N_m\geqslant 1$ be integers such that $N=N_1+\dots+N_m$. We consider the scalar product
\begin{equation}
O_{N_1,\dots,N_m} = \langle \psi_{N}|\left(|\psi_{N_1}\rangle\otimes\cdots \otimes |\psi_{N_m}\rangle\right).
\end{equation}
The magnetisations of the ground-state vectors imply that this scalar product trivially vanishes if more than one of the integers $N_1,\dots, N_m$ is odd. We have obtained the scalar products for the non-trivial cases up to $N=12$ sites through the exact computation of the ground-state vector with \textsc{Mathematica}. Our results suggest the following conjecture in terms of $F_N=F_N(x,\alpha)$:
\begin{conjecture}
  \label{conj:Overlaps} If at most one of the integers $N_1,\dots,N_m$ is odd, then we have   \begin{equation}
    \label{eqn:ConjectureO}
    O_{N_1,\dots,N_m} =  x^n F_{N}(x,x^{-1})\, \prod_{i=1}^m \gamma_{N_i},
  \end{equation}
  where $\gamma_{2k} = A_{\mathrm{V}}(2k+1)$ and $\gamma_{2k+1}=N_8(2k+2)$ for each $k\geqslant 0$.
\end{conjecture}
The factorisation of $O_{N_1,\dots,N_m}$ into a product of determinants is unexpectedly simple. Its dependence on $x$ is completely determined by the sum $N$, but not by the individual choices of the integers $N_1,\dots,N_m$. Moreover, the factorisation implies that the scalar product is invariant under any permutation of the integers $N_1,\dots,N_m$, which is a surprising and counter-intuitive property.

For the supersymmetric point $x=1$, the proof of the conjecture follows from \eqref{eqn:AlternatingEntriesSUSY} and \eqref{eqn:SquareNormSUSY}, as well as a simple factorisation of the scalar product $O_{N_1,\dots,N_m}$ that we established in Theorem 5.9 of \cite{hagendorf:17}. For arbitrary $x$, the proof and the evaluation of the logarithmic bipartite fidelity for large $N$ are beyond the scope of this article.

%%%%%%%%%%%%%%%%%%%%%%%%%%%%%%%%%%%%%%%%%%%%%%%%%%%%%%%%%%%%%%%%%%%%%%%%%%%%%%%%%%%%%%%%%%%%%
\section{A solution to the boundary quantum Knizhnik-Zamolod{\-}chikov equations}
\label{sec:bqKZ}
%%%%%%%%%%%%%%%%%%%%%%%%%%%%%%%%%%%%%%%%%%%%%%%%%%%%%%%%%%%%%%%%%%%%%%%%%%%%%%%%%%%%%%%%%%%%%

In this section, we introduce and analyse a vector $|\Psi_N\rangle \in V^N$ that depends on $N$ complex variables $z_1,\dots,z_N$. We define its components through multiple contour integrals in \cref{sec:IntegralFormulas}. We use the contour integrals in \cref{sec:ExchangeRelations} to show that the vector obeys the \textit{exchange relations}. In \cref{sec:ReflectionEquations}, we show that it furthermore satisfies two \textit{reflection relations}.
The exchange and reflection relations imply that the vector is a solution to the so-called \textit{boundary quantum Knizhnik-Zamolodchikov equations}. In \cref{sec:Polynomiality}, we prove that the components of $|\Psi_N\rangle$ are Laurent polynomials in $z_1,\dots,z_N$ and compute their degrees. We use their properties in \cref{sec:Parity} to obtain the vector's behaviour under a parity transformation.

%%%%%%%%%%%%%%%%%%%%%%%%%%%%%%%%%%%%%%%%%%%%%%%%%%%%%%%%%%%%%%%%%%%%%%%%%%%%%%%%%%%%%%%%%%%%%
\subsection{Integral formulas}
\label{sec:IntegralFormulas}
%%%%%%%%%%%%%%%%%%%%%%%%%%%%%%%%%%%%%%%%%%%%%%%%%%%%%%%%%%%%%%%%%%%%%%%%%%%%%%%%%%%%%%%%%%%%%
Throughout this article, we systematically use the notation $[z] = z-z^{-1}$, as well as the integers $n$ and $\bar n$, defined in \eqref{eqn:Defnnbar}. Let $a_1,\dots,a_n$ be integers that satisfy $1\leqslant a_1 < \dots < a_n \leqslant N$. Inspired by \cite{difrancesco:07_2,degier:09,fonseca:09}, we define the multiple contour integral
\begin{multline}
  \label{eqn:DefPsiCI}
  (\Psi_N)_{a_1,\dots,a_n}(z_1,\dots, z_N) =(-[q])^n \prod_{1\leqslant i < j \leqslant N}\left[qz_j/z_i\right]\left[q^2z_iz_j\right]\\
  \times \oint \cdots \oint \prod_{\ell=1}^n\frac{\diff w_\ell}{\pi \i w_\ell} \,\Xi_{a_1,\dots,a_n}(w_1,\dots,w_n|z_1,\dots,z_N),
\end{multline}
where $z_1,\dots,z_N$ are generic complex numbers. Moreover, $q \in \mathbb C\backslash \{0,1,-1\}$ is a complex parameter.
The integrand contains the meromorphic function
\begin{multline}
  \label{eqn:DefXi}
 \Xi_{a_1,\dots,a_n}(w_1,\dots,w_n|z_1,\dots,z_N)\\
  =\frac{\prod_{1\leqslant i < j \leqslant n}[q w_j/w_i][w_i/w_j]
 [qw_iw_j]\prod_{1\leqslant i \leqslant j \leqslant n}[q^2w_iw_j]
  \prod_{i=1}^n[\beta w_i]}{\prod_{i=1}^{n}\left(\prod_{j=1}^{a_i}[z_j/w_i]\prod_{j=a_i}^{N}[qz_j/w_i]\prod_{j=1}^N[q^2w_iz_j]\right)},
\end{multline}
where $\beta \in \mathbb C\backslash \{0\}$. The integration contour of $w_i$ in \eqref{eqn:DefPsiCI} is a collection of positively-oriented curves surrounding the poles $w_i=z_j$, $j=1,\dots,N$. These curves do not surround the other singularities $w_i=0,-z_j,\pm q z_j,\pm q^{-2}z_j^{-1}$, where $j=1,\dots,N$. Similarly, let $b_1,\dots,b_{\bar n}$ be integers with $1\leqslant b_1 < \dots < b_{\bar n} \leqslant N$. We define the multiple contour integral
\begin{multline}
  \label{eqn:DefPsiBarCI}
    (\overline{\Psi}_N)_{b_1,\dots,b_{\bar{n}}} (z_1,\dots, z_N)%
= [q]^{\bar n}\prod_{1\leqslant i < j \leqslant N}\left[qz_j/z_i\right]\left[qz_iz_j\right]\prod_{i=1}^N [\beta z_i]\\
 \times \oint \cdots \oint \prod_{\ell=1}^{\bar{n}}\frac{\diff w_\ell}{\pi \i w_\ell} \,\overline\Xi_{b_1,\dots,b_{\bar{n}}}(w_1,\dots,w_{\bar{n}}|z_1,\dots,z_N),
\end{multline}
whose integrand contains the meromorphic function
\begin{multline}
  \label{eqn:DefXiBar}
  \overline{\Xi}_{b_1,\dots,b_{\bar{n}}}(w_1,\dots,w_{\bar{n}}|z_1,\dots,z_N)\\
  =\frac{\prod_{1\leqslant i < j \leqslant \bar{n}}[qw_j/w_i][w_i/w_j][q^2w_iw_j]\prod_{1\leqslant i \leqslant j \leqslant \bar{n}}[qw_iw_j]}{\prod_{i=1}^{\bar{n}}\left(\prod_{j=1}^{b_i}[q w_i/z_j]\prod_{j=b_i}^{N}[w_i/z_j]\prod_{j=1}^N[qw_iz_j]\right)\prod_{i=1}^{\bar{n}}[\beta w_i]}.
\end{multline}
The integration contour of $w_i$ in \eqref{eqn:DefPsiBarCI} is a collection of positively-oriented curves surrounding the poles $w_i=z_j$, but not the singularities located at $w_i=0,-z_j,\pm q^{-1} z_j,\pm q^{-1}z_j^{-1},\pm \beta^{-1}$, where $j=1,\dots,N$.

It is possible to apply the residue theorem and compute an explicit (combinatorial) formula for the multiple contour integrals defined in \eqref{eqn:DefPsiCI} and \eqref{eqn:DefPsiBarCI}. This  
formula is, however, not quite useful for explicit computations. It only simplifies in the two cases $a_i = i$ and $b_i = \bar n+i$. In these cases, we obtain:
\begin{proposition}
\label{prop:SpecialComponent}
For each $N\geqslant 2$, we have
\begin{multline}
  \label{eqn:SpecialComponent}
  (\Psi_N)_{1,\dots,n}(z_1,\dots, z_N)=  (\overline{\Psi}_N)_{\bar n+1,\dots,N}(z_1,\dots, z_N) \\
   =\prod_{i=1}^n[\beta z_i]\prod_{1\leqslant i < j \leqslant n}[qz_iz_j][qz_j/z_i]\prod_{n+1\leqslant i<j\leqslant N} [qz_j/z_i|[q^2z_iz_j].
\end{multline}
\begin{proof}
We only sketch the evaluation of $(\Psi_N)_{1,\dots,n}(z_1,\dots, z_N)$. To this end, we iteratively compute the contour integrals \eqref{eqn:DefPsiCI} with respect to $w_1,\dots,w_n$ for $a_i=i$. We observe that the only pole that contributes to the contour integral with respect to $w_\ell$ is $z_\ell$. The evaluation of its residue leads to \eqref{eqn:SpecialComponent}. The computation of $(\overline{\Psi}_N)_{\bar n+1,\dots,N}(z_1,\dots, z_N)$ is similar.
\end{proof}
\end{proposition}

We now introduce the two vectors $|\Psi_N\rangle=|\Psi_N(z_1,\dots,z_N)\rangle$ and $|\overline\Psi_N\rangle=|\overline\Psi_N(z_1,\dots,z_N)\rangle$. For $N=1$, they are given by $|{\Psi_1}\rangle=|{\overline\Psi_1}\rangle=|{\uparrow}\rangle$. For $N\geqslant 2$, we use \eqref{eqn:DefPsiCI} and \eqref{eqn:DefPsiBarCI} to define them as
\begin{alignat}{3}
  \label{eqn:DefPsi}
  |\Psi_N(z_1,\dots, z_N)\rangle &= \sum_{1\leqslant a_1 < \cdots < a_n \leqslant N} 
  &&(\Psi_N)_{a_1,\dots,a_n}(z_1,\dots, z_N)\,
  &&|\uparrow\cdots \uparrow\underset{a_1}{\downarrow}\uparrow \quad \cdots \quad\uparrow
  \underset{a_n}{\downarrow}\uparrow \cdots \uparrow\rangle,\\
  \label{eqn:DefPsiBar}
   |\overline\Psi_N(z_1,\dots, z_N) \rangle &= \sum_{1\leqslant b_1 < \cdots < b_{\bar n} \leqslant N} 
   &&(\overline{\Psi}_N)_{b_1,\dots,b_{\bar n}} (z_1,\dots, z_N)
   &&|\downarrow\cdots \downarrow\underset{b_1}{\uparrow}\downarrow \quad \cdots \quad \downarrow
   \underset{b_{\bar{n}}}{\uparrow}\downarrow \cdots \downarrow\rangle.
\end{alignat}
Here and in the following, we label the components of a vector in terms of the positions of the up or down spins of the associated spin configuration
(as opposed to the labelling by the spin configurations used in \cref{sec:SpinChainGS}). Moreover, we only write out the dependence on $z_1,\dots,z_N$ if necessary.

It follows from \cref{prop:SpecialComponent} that the vectors $|\Psi_N\rangle$ and $|\overline \Psi_N\rangle$ do not identically vanish. The purpose of the following section is to investigate their properties.

%%%%%%%%%%%%%%%%%%%%%%%%%%%%%%%%%%%%%%%%%%%%%%%%%%%%%%%%%%%%%%%%%%%%%%%%%%%%%%%%%%%%%%%%%%%%%
\subsection{The exchange relations}
\label{sec:ExchangeRelations}
%%%%%%%%%%%%%%%%%%%%%%%%%%%%%%%%%%%%%%%%%%%%%%%%%%%%%%%%%%%%%%%%%%%%%%%%%%%%%%%%%%%%%%%%%%%%%

To formulate the exchange relations, we introduce the $R$-matrix of the six-vertex model. It is an operator on $V^2$ that acts on the canonical basis $\{|{\uparrow\uparrow}\rangle,|{\uparrow\downarrow}\rangle,|{\downarrow\uparrow}\rangle,|{\downarrow\downarrow}\rangle\}$ as the matrix
\begin{subequations}
\label{eqn:DefR}
\begin{equation}
  R(z) =
  \begin{pmatrix}
    a(z) & 0 & 0 & 0\\
    0 & b(z) & c(z) & 0\\
    0 & c(z) & b(z) & 0\\
    0 & 0 & 0 & a(z)
  \end{pmatrix}.
\end{equation}
In this article, we choose the following parameterisation for the entries of $R(z)$:
\begin{equation}
  a(z) = [q z]/[q/z], \quad b(z) = [z]/[q/z], \quad c(z) = [q]/[q/z].
\end{equation}
\end{subequations}
The $R$-matrix obeys the Yang-Baxter equation. On $V^3$, it is given by
\begin{equation}
  \label{eqn:YBE}
  R_{1,2}(z/w)R_{1,3}(z)R_{2,3}(w) = R_{2,3}(w)R_{1,3}(z)R_{1,2}(z/w).
\end{equation}
Moreover, we have have $R(1)= \piJ $, where $\piJ$ is the permutation operator acting according to $\piJ(|v\rangle \otimes |w\rangle) = |w\rangle \otimes |v\rangle$ for any $|v\rangle,|w\rangle \in V$. Using this operator, we define the $\check R$-matrix by
\begin{equation}
\check R(z) = \piJ R(z).
\end{equation}
It follows from \eqref{eqn:YBE} that it obeys the braid version of the Yang-Baxter equation
\begin{equation}
  \label{eqn:BraidYBE}
  \check R_{1,2}(z/w)\check R_{2,3}(z)\check R_{1,2}(w) = \check R_{2,3}(w)\check R_{1,2}(z)\check R_{2,3}(z/w).
\end{equation}

%%%%%%%%%%%%%%%%%%%%%%%%%%%%%%%%%%%%%%%%%%%%%%%%%%%%%%%%%%%%%%%%%%%%%%%%%%%%%%%%%%%%%%%%%%%%%
\subsubsection*{The exchange relations}
%%%%%%%%%%%%%%%%%%%%%%%%%%%%%%%%%%%%%%%%%%%%%%%%%%%%%%%%%%%%%%%%%%%%%%%%%%%%%%%%%%%%%%%%%%%%%

We say that a vector $|\Phi\rangle=|\Phi(z_1,\dots,z_N)\rangle \in V^N,\,N\geqslant 2,$ that depends on the complex numbers $z_1,\dots,z_N$, obeys the \textit{exchange relations} if
\begin{equation}
  \check R_{i,i+1}(z_i/z_{i+1})|\Phi(\dots,z_i,z_{i+1},\dots)\rangle = |\Phi(\dots,z_{i+1},z_i,\dots)\rangle,
  \label{eqn:ExchangeRelations}
\end{equation}
for each $i=1,\dots,N-1$. The compatibility of these equations follows from \eqref{eqn:BraidYBE}, extended to $V^N$.
\begin{proposition}
  \label{prop:Exchange}
  For each $N\geqslant 2$, the  vectors $|\Psi_N\rangle$ and $|\overline\Psi_N\rangle$ obey the exchange relations \eqref{eqn:ExchangeRelations}.
\end{proposition}
\begin{proof}
The proofs of the exchange relations for $|\Psi_N\rangle$ and $|\overline \Psi_N\rangle$ are similar and follow the lines of \cite{razumov:07}. Hence, we focus on $|\Psi_N\rangle$. To prove that it obeys the exchange relations, we consider integers $a_1,\dots,a_n$ with $1\leqslant a_1<a_2<\cdots < a_n \leqslant N$ and an integer $1\leqslant i \leqslant N-1$. We examine four cases, depending on whether $i$ and $i+1$ belong to $ \{a_1,\dots,a_n\}$. In this proof, we use $\Xi_{a_1,\dots,a_n}(z_1,\dots,z_N)$ as a shorthand notation for $\Xi_{a_1,\dots,a_n}(w_1,\dots, w_n|z_1,\dots,z_N)$.

\medskip

\noindent \textit{Case 1: $i,i+1 \notin \{a_1,\dots,a_n\}$.} In this case, we note that $\Xi_{a_1,\dots,a_n}(\dots,z_i,z_{i+1},\dots )$ is symmetric under the exchange of $z_i$ and $z_{i+1}$. We combine this observation with \eqref{eqn:DefPsiCI} and conclude that the quotient of $(\Psi_N)_{a_1,\dots,a_n}(\dots,z_i,z_{i+1},\dots)$ and $[qz_{i+1}/z_i]$ is symmetric under the exchange of $z_i$ and $z_{i+1}$. Hence, we find the relation
\begin{equation}
  \label{eqn:ExchangeComponents12}
  \frac{[qz_i/z_{i+1}]}{[q z_{i+1}/z_i]}(\Psi_N)_{a_1,\dots,a_n}(\dots,z_i,z_{i+1},\dots)=(\Psi_N)_{a_1,\dots,a_n}(\dots,z_{i+1},z_i,\dots).
\end{equation}

\noindent \textit{Case 2: $i,i+1 \in \{a_1,\dots,a_n\}$.} Let $1\leqslant \ell \leqslant n-1$ be the integer such that $a_\ell = i$. We have
\begin{multline}
  \Xi_{a_1,\dots,a_n}(\dots,z_{i+1},z_{i},\dots) -\Xi_{a_1,\dots,a_n}(\dots,z_i,z_{i+1},\dots)\\
  = \frac{[q w_\ell/w_{\ell+1}][z_i/z_{i+1}]}{[q z_{i}/w_{\ell+1}][z_{i+1}/w_{\ell}]} 
  \Xi_{a_1,\dots,a_n}(\dots,z_i,z_{i+1},\dots)
\end{multline}
The inspection of \eqref{eqn:DefXi} allows us to conclude that this difference is antisymmetric under the exchange of $w_\ell$ and $w_{\ell+1}$. Hence, the multiple contour integral over the difference vanishes. It follows that the quotient of $(\Psi_N)_{a_1,\dots,a_n}(\dots,z_i,z_{i+1},\dots)$ and $[qz_{i+1}/z_i]$ is symmetric under the exchange of $z_i$ and $z_{i+1}$. Therefore, the relation \eqref{eqn:ExchangeComponents12} holds in this case, too.

\medskip

\noindent \textit{Case 3: $i \in \{a_1,\dots,a_n\}$ and $i+1\notin \{a_1,\dots,a_n\}$.} Let $1\leqslant \ell \leqslant n-1$ be the integer such that $a_\ell = i$. We find 
\begin{align}
  \Xi_{a_1,\dots,i+1,\dots,a_n}(\dots,z_i,z_{i+1},\dots ) & = \frac{[qz_i/w_\ell]}{[z_{i+1}/w_\ell]}\Xi_{a_1,\dots,i,\dots,a_n}(\dots,z_i,z_{i+1},\dots ),\\
  \Xi_{a_1,\dots,i,\dots,a_n}(\dots,z_{i+1},z_i,\dots ) & = \frac{[z_i/w_\ell]}{[z_{i+1}/w_\ell]}\Xi_{a_1,\dots,i,\dots,a_n}(\dots,z_i,z_{i+1},\dots ).
\end{align}
We combine these relations into the equality
\begin{multline}
  [q] \Xi_{a_1,\dots,i,\dots,a_n}(\dots,z_i,z_{i+1},\dots)+[z_i/z_{i+1}]\Xi_{a_1,\dots,i+1,\dots,a_n}(\dots,z_i,z_{i+1},\dots)\\
   =[qz_{i}/z_{i+1}]\Xi_{a_1,\dots,i,\dots,a_n}(\dots,z_{i+1},z_{i},\dots).
\end{multline}
Using this equality, it is straightforward to show that
\begin{multline}
  \frac{[q]}{[qz_{i+1}/z_i]} (\Psi_N)_{a_1,\dots,i,\dots,a_n}(\dots,z_i,z_{i+1},\dots)+\frac{[z_i/z_{i+1}]}{{[qz_{i+1}/z_i]}}(\Psi_N)_{a_1,\dots,i+1,\dots,a_n}(\dots,z_i,z_{i+1},\dots)\\
 =(\Psi_N)_{a_1,\dots,i,\dots,a_n}(\dots,z_{i+1},z_{i},\dots).\label{eqn:ExchangeComponents3}
\end{multline}

\noindent \textit{Case 4: $i \notin \{a_1,\dots,a_n\}$ and $i+1\in \{a_1,\dots,a_n\}$} The analysis of this case is very similar to the previous one. One obtains the relation
\begin{multline}
   \frac{[q]}{[qz_{i+1}/z_i]} (\Psi_N)_{a_1,\dots,i+1,\dots,a_n}(\dots,z_i,z_{i+1},\dots)+\frac{[z_i/z_{i+1}]}{{[qz_{i+1}/z_i]}}(\Psi_N)_{a_1,\dots,i,\dots,a_n}(\dots,z_i,z_{i+1},\dots)\\
  =(\Psi_N)_{a_1,\dots,i+1,\dots,a_n}(\dots,z_{i+1},z_{i},\dots).  \label{eqn:ExchangeComponents4}
\end{multline}

To conclude, we note that \eqref{eqn:ExchangeComponents12}, \eqref{eqn:ExchangeComponents3} and \eqref{eqn:ExchangeComponents4} are equal to the exchange relations, written for the components of $|\Psi_N\rangle$, which ends the proof.
\end{proof}

%%%%%%%%%%%%%%%%%%%%%%%%%%%%%%%%%%%%%%%%%%%%%%%%%%%%%%%%%%%%%%%%%%%%%%%%%%%%%%%%%%%%%%%%%%%%%
\subsubsection*{Properties of solutions to the exchange relations}
%%%%%%%%%%%%%%%%%%%%%%%%%%%%%%%%%%%%%%%%%%%%%%%%%%%%%%%%%%%%%%%%%%%%%%%%%%%%%%%%%%%%%%%%%%%%%

We now investigate a few simple properties of a vector $|\Phi\rangle=|\Phi(z_1,\dots,z_N)\rangle \in V^N,\, N\geqslant 2,$ with magnetisation $\mu=(\bar n - n)/2$ that obeys the exchange relations \eqref{eqn:ExchangeRelations}. We define its components through the expansion
\begin{equation}
  \label{eqn:DefPhi}
  |\Phi\rangle = \sum_{1\leqslant a_1 < \cdots < a_n \leqslant N} \Phi_{a_1,\dots,a_n}(z_1,\dots,z_N)|{\uparrow}\cdots \uparrow\underset{a_1}{\downarrow}\uparrow \quad \cdots \quad\uparrow\underset{a_n}{\downarrow}\uparrow \cdots \uparrow\rangle.
\end{equation}
Following the proof of \cref{prop:Exchange}, we rewrite the exchange relations in terms of the components of the vector.  There are four different cases. The next lemma addresses two of them.
\begin{lemma}
  \label{lem:ExchangeComponents1}
  Let $a_1,\dots,a_n$ be integers such that $1\leqslant a_1 < \cdots < a_n \leqslant N$. If $1\leqslant i \leqslant N-1$ is an integer such that either $i,i+1 \notin \{a_1,\dots,a_n\}$ or $i,i+1 \in \{a_1,\dots,a_n\}$ then
  \begin{equation}
    \Phi_{a_1,\dots,a_n}(\dots,z_i,z_{i+1},\dots) = [qz_{i+1}/z_i] \bar \Phi_{a_1,\dots,a_n}(\dots,z_i,z_{i+1},\dots),
  \end{equation}
  where $\bar \Phi_{a_1,\dots,a_n}(\dots,z_i,z_{i+1},\dots)$ is symmetric under the exchange of $z_i$ and $z_{i+1}$.
\end{lemma}
For the two other cases, we introduce the divided difference operator $\delta$. It acts on a function $f$ of two complex variables $z,w$ according to
\begin{equation}\label{eqn:DefDividedDiff}
  \delta f(z,w) = \frac{[q w/z]f(w,z)-[q]f(z,w)}{[z/w]}.
\end{equation}
More generally, for a function $f$ depending on $z_1,\dots,z_N$, we write $\delta_i f$ for the action of the divided difference operator $\delta$ on $f$ with $z=z_i$ and $w=z_{i+1}$.

\begin{lemma}
  \label{lem:ExchangeComponents2}
  Let $a_1,\dots,a_n$ be integers such that $1\leqslant a_1 < \cdots < a_n \leqslant N$. If $1\leqslant i \leqslant N-1$ is an integer such that $i\in \{a_1,\dots,a_n\}$ and $i+1 \notin \{a_1,\dots,a_n\}$, then
  \begin{align}
    \Phi_{a_1,\dots,i+1,\dots,a_n}(\dots,z_i,z_{i+1},\dots) &= \delta_i \Phi_{a_1,\dots,i,\dots,a_n}(\dots,z_i,z_{i+1},\dots),\\
    \Phi_{a_1,\dots,i,\dots,a_n}(\dots,z_i,z_{i+1},\dots)& = \delta_i \Phi_{a_1,\dots,i+1,\dots,a_n}(\dots,z_i,z_{i+1},\dots).
   \end{align}
\end{lemma}

This lemma allows us to prove the following useful property of the solutions to the exchange relations:
\begin{proposition}
  \label{prop:ExchangeProperty}
Suppose that there are integers $\bar a_1,\dots,\bar a_n$ with $1\leqslant \bar a_1<\cdots<\bar a_n\leqslant N$ such that $\Phi_{\bar a_1,\dots,\bar a_n}(z_1,\dots,z_N)$ vanishes identically then the vector $|\Phi\rangle$ vanishes identically.
\end{proposition}
\begin{proof}
\Cref{lem:ExchangeComponents2} allows us to write for all integers $a_1,\dots,a_n$ with $1\leqslant a_1< \dots < a_n\leqslant N$ the relation
\begin{equation}
  \label{eqn:PhiComponentFromSpecialComponent}
  \Phi_{a_1,\dots,a_n}(z_1,\dots,z_N) = 
 \left(\prod_{i=1,\dots,n} \prod_{j=i,\dots,a_i-1}^\curvearrowleft  \delta_j\right) 
  \Phi_{1,\dots,n}(z_1,\dots,z_N).
\end{equation}
Here, $\curvearrowleft$ indicates that the products of operators are taken in reverse order. Since the component $\Phi_{\bar a_1,\dots,\bar a_n}(z_1,\dots,z_N)$ vanishes identically, we find
\begin{align}
  0 &= \left(\prod_{i=1,\dots,n}^\curvearrowleft \prod_{j=i,\dots,\bar a_{i}-1}\delta_j\right)\Phi_{\bar a_1,\dots,\bar a_n}(z_1,\dots,z_N)\\
     &= \left(\prod_{i=1,\dots,n}^\curvearrowleft \prod_{j=i,\dots,\bar a_{i}-1}\delta_j\right)\left(\prod_{i=1,\dots, n} \prod_{j=i,\dots,\bar a_i-1} ^\curvearrowleft \delta_j\right)\Phi_{1,\dots, n}(z_1,\dots,z_N). 
     \label{eqn:IntermediateXi}
 \end{align}
One checks that for each $j=1,\dots,N-1$, the divided difference operator $\delta_j$ has the property $\delta_j^2 =\textup{id}$. We apply this property in \eqref{eqn:IntermediateXi} and conclude that $\Phi_{1,\dots, n}(z_1,\dots,z_N)$ vanishes identically. It follows from \eqref{eqn:PhiComponentFromSpecialComponent} that all components vanish identically.
\end{proof}

We now show that the two vectors defined in \cref{sec:IntegralFormulas} are equal. This equality allows us to limit our investigation to $|\Psi_N\rangle$. For each of its (non-trivial) components, we have two different multiple contour integral formulas.
\begin{proposition}
  \label{prop:Equality}
  We have $|\overline \Psi_N\rangle = |\Psi_N\rangle$.
\end{proposition}
\begin{proof}
For $N=1$, the proposition holds by the definition of the vectors. Hence, we consider the difference $|\Phi\rangle = |\overline \Psi_N\rangle-|\Psi_N\rangle$ for $N\geqslant 2$. It follows from \eqref{eqn:DefPsi} and \eqref{eqn:DefPsiBar} that this vector is of the form \eqref{eqn:DefPhi}. Moreover,  \cref{prop:Exchange} implies that it obeys the exchange relations. It has the component 
  \begin{equation}
     \Phi_{1,\dots,n}(z_1,\dots,z_N) = (\overline \Psi_N)_{\bar n+1,\dots,N}(z_1,\dots,z_N)-(\Psi_N)_{1,\dots,n}(z_1,\dots,z_N) = 0,
  \end{equation}
  as follows from \cref{prop:SpecialComponent}.
  By virtue of \cref{prop:ExchangeProperty}, we conclude that $|\Phi\rangle$ vanishes identically.
\end{proof}

%%%%%%%%%%%%%%%%%%%%%%%%%%%%%%%%%%%%%%%%%%%%%%%%%%%%%%%%%%%%%%%%%%%%%%%%%%%%%%%%%%%%%%%%%%%%%
\subsection{The reflection relations}
\label{sec:ReflectionEquations}
%%%%%%%%%%%%%%%%%%%%%%%%%%%%%%%%%%%%%%%%%%%%%%%%%%%%%%%%%%%%%%%%%%%%%%%%%%%%%%%%%%%%%%%%%%%%%
The reflection relations for the vector $|\Psi_N\rangle$ are written in terms of a $K$-matrix. It is an operator $K(z)$ on $V$ that solves the boundary Yang-Baxter equation for the six-vertex model \cite{sklyanin:88},
\begin{equation}
  \label{eqn:bYBE}
  R_{1,2}(z/w)K_1(z)R_{1,2}(zw)K_2(w)= K_2(w)R_{1,2}(zw)K_1(z)R_{1,2}(z/w).
\end{equation}
The most general solution of this equation can be found in \cite{vega:94}. In this article, we consider Cherednik's diagonal solution $K(z)=K(z;\beta)$ \cite{cherednik:92}. It acts on the canonical basis $\{|{\uparrow}\rangle, |{\downarrow}\rangle\}$ as the matrix
\begin{equation}
 \label{eqn:DefK}
  K(z;\beta) = 	\begin{pmatrix}
   					 1 & 0\\
  					 0 & [\beta z]/[\beta/z]
				\end{pmatrix},
\end{equation}
where $\beta$ is a non-zero complex number.

%%%%%%%%%%%%%%%%%%%%%%%%%%%%%%%%%%%%%%%%%%%%%%%%%%%%%%%%%%%%%%%%%%%%%%%%%%%%%%%%%%%%%%%%%%%%%
\subsubsection*{The reflection relations}
%%%%%%%%%%%%%%%%%%%%%%%%%%%%%%%%%%%%%%%%%%%%%%%%%%%%%%%%%%%%%%%%%%%%%%%%%%%%%%%%%%%%%%%%%%%%%
A vector $|\Phi\rangle=|\Phi(z_1,\dots,z_N)\rangle\in V^N$ that depends on the complex numbers $z_1,\dots,z_N$ obeys the reflection relations if
\begin{align}
  K_1(z_1^{-1};\beta)|\Phi(z_1,\dots,z_N)\rangle &=|\Phi(z_1^{-1},\dots,z_N)\rangle,
  \label{eqn:ReflectionLeft} \\
  K_N(s z_N;s\bar \beta)|\Phi(z_1,\dots,z_N)\rangle & =|\Phi(z_1,\dots,s^{-2}z_N^{-1})\rangle.
  \label{eqn:ReflectionRight}
\end{align}
Here, $s$ and $\bar \beta$ are two complex parameters. Throughout this section, we assume that they obey the relations
\begin{equation}
  \label{eqn:RelationSQ}
  s^4 = q^6,
\end{equation}
and
\begin{equation}
  \label{eqn:BetaBarBeta}
  \bar \beta^2\beta^2 q^2 = 1.
\end{equation}
\begin{proposition}
\label{prop:Reflection}
For each $N\geqslant 1$, the vector $|\Psi_N\rangle$ obeys the reflection relations.
\end{proposition}
\begin{proof} 
  The case $N=1$ is trivial. Hence, we consider $N\geqslant 2$. We present the proof of the second reflection relation \eqref{eqn:ReflectionRight} for $|\Phi\rangle = |\Psi_N\rangle$ in detail. To this end, we choose integers $a_1,\dots,a_n$ with $1\leqslant a_1<\dots < a_n \leqslant N$, and establish the two equations
\begin{align}
    \label{eqn:ReflectionRight1}
    (\Psi_N)_{a_1,\dots,a_{n}}(z_1,\dots,z_N) 
    				&= (\Psi_N)_{a_1,\dots,a_n}(z_1,\dots,s^{-2}z_N^{-1}),&& \text{if} \quad a_n < N,\\
\intertext{and}
    \label{eqn:ReflectionRight2}
    \frac{[s^2\bar\beta z_N]}{[\bar \beta z_N^{-1}]}(\Psi_N)_{a_1,\dots,a_{n}}(z_1,\dots,z_N)
    				&= (\Psi_N)_{a_1,\dots,a_{n}}(z_1,\dots,s^{-2}z_N^{-1}),&&\text{if} \quad a_n = N.
\end{align}
These two equations are equivalent to the second reflection relation.
  \medskip
  
  \noindent \textit{Case 1:} $a_n<N$. Using \eqref{eqn:SpecialComponent} and \eqref{eqn:RelationSQ}, it is straightforward to show that
  \begin{equation}
    (\Psi_N)_{1,\dots,n}(z_1,\dots,z_N) = (\Psi_N)_{1,\dots,n}(z_1,\dots,s^{-2}z_N^{-1}).
  \end{equation}
  For $n=1$, there is nothing left to prove. For $n\geqslant 2$, we apply \cref{lem:ExchangeComponents2} to obtain \eqref{eqn:ReflectionRight1}.
  
  \medskip
  \noindent \textit{Case 2:} $a_N=N$. We consider the difference
  \begin{multline}
    \Delta(z_1,\dots,z_N) = [s^2\bar\beta z_N](\Psi_N)_{1,\dots,n-1,N}(z_1,\dots,z_N)\\
    - [\bar \beta z_N^{-1}](\Psi_N)_{1,\dots,n-1,N}(z_1,\dots,s^{-2}z_N^{-1}).
    \label{eqn:Difference}
  \end{multline}
  We compute it by means of the contour integral formula \eqref{eqn:DefPsiCI}. The integrations with respect to $w_1,\dots,w_{n-1}$ are straightforward. Using \eqref{eqn:RelationSQ}, we find
  \begin{equation}
    \Delta(z_1,\dots,z_N) =p(z_1,\dots,z_N)\oint_{C}\diff w_n\,f(w_n),
    \label{eqn:DiffSimple}
  \end{equation}
  where
   \begin{equation}
    p(z_1,\dots,z_N) = -[q][s^2 z_{N}^2]
     \prod_{1\leqslant i<j\leqslant n-1}[qz_j/z_i][qz_iz_j]\prod_{n\leqslant i<j\leqslant N}[qz_j/z_i][q^2z_iz_j]\prod_{i=1}^{n-1}[\beta z_i],
  \end{equation}
  and
  \begin{equation}
    f(w) = \frac{[\bar \beta q^3w][\beta w][q^2w^2]\prod_{i=1}^{n-1}[qw/z_i][q wz_i]}{\i \pi w \prod_{i=n}^{N+1}[z_i/w][q^2 w z_i]}.
  \end{equation}
 Here, we abbreviated $z_{N+1}=s^{-2}z_N^{-1}$. The integration contour $C$ is a collection of positively-oriented curves around the simple poles $z_n,\dots,z_{N+1}$, but not around any other pole of $f$.
  
  We now analyse the contour integral in \eqref{eqn:DiffSimple}. To this end, we make three simple observations. First, $f$ has a removable singularity at $w=0$. All other singularities of $f$ are simple poles, located at $w=z_n,\dots,z_{N+1}$ and $w=\varphi_i(z_n),\dots,\varphi_i(z_{N+1})$, $i=1,2,3$, where 
  \begin{equation}
     \varphi_1(w)=-w, \quad \varphi_2(w)=q^{-2}w^{-1}, \quad \varphi_3(w)=-q^{-2}w^{-1}.
   \end{equation} 
   Second, we note that $f$ obeys $\varphi_i'(w)f(\varphi_i(w)) = f(w)$. This is trivial for $i=1$. For $i=2,3$, it follows from \eqref{eqn:BetaBarBeta}. These properties of $f$ allow us to write
  \begin{equation}
    \label{eqn:ContourIntegralInvariance}
    \oint_{\varphi_i(C)} \diff w_n\, f(w_n) = \oint_{C} \diff w_n\, f(w_n) , \quad i=1,2,3.
  \end{equation}
Third, we note that $f$ tends to zero at infinity and that it has no residue at infinity. This allows us to push the integration contour $C$ in \eqref{eqn:DiffSimple} to infinity, which results in
\begin{equation}
  \oint_{C}\diff w_n\,f(w_n) = - \sum_{i=1}^{3}\oint_{\varphi_i(C)}\diff w_n\,f(w_n) = -3\oint_{C} \diff w_n\, f(w_n).
\end{equation}
Here, we used \eqref{eqn:ContourIntegralInvariance}.
We conclude from this equation that the contour integral vanishes. Hence, $\Delta(z_1,\dots,z_N)$ vanishes, and therefore
\begin{equation}
\frac{[s^2\bar\beta z_N]}{[\bar \beta z_N^{-1}]}(\Psi_N)_{1,\dots,n-1,N}(z_1,\dots,z_N)= (\Psi_N)_{1,\dots,n-1,N}(z_1,\dots,s^{-2}z_N^{-1}).
\end{equation}
For $n=1$, there is nothing left to prove. For $n\geqslant 2$, we use \cref{lem:ExchangeComponents2} to obtain \eqref{eqn:ReflectionRight2}. This ends the proof of \eqref{eqn:ReflectionRight}.

Finally, we comment on the proof of \eqref{eqn:ReflectionLeft}. It amounts to establishing the two relations
\begin{align}
  \label{eqn:ReflectionLeft1}
  \frac{[\beta/z_1]}{[\beta z_1]}(\Psi_N)_{a_1,\dots, a_n}(z_1,\dots,z_N) &= (\Psi_N)_{a_1,\dots,a_n}(z_1^{-1},\dots,z_N), \quad \text{if} \quad a_1 =1,\\
  \label{eqn:ReflectionLeft2}
  (\Psi_N)_{a_1,\dots,a_N}(z_1,\dots,z_N) &= (\Psi_N)_{a_1,\dots,a_{n}}(z_1^{-1},\dots,z_N), \quad \text{if} \quad a_1>1.
\end{align}
The first relation is easily proven for the special choice $a_i=i$, using the special component \eqref{eqn:SpecialComponent}. The general relation \eqref{eqn:ReflectionLeft1} then follows from \cref{lem:ExchangeComponents2}. The proof of the second relation is based on showing that the difference
\begin{equation}
  \bar \Delta(z_1,\dots,z_N) = (\Psi_N)_{2,\dots,n+1}(z_1,\dots,z_N)- (\Psi_N)_{2,\dots,n+1}(z_1^{-1},\dots,z_N)
\end{equation}
vanishes. \Cref{prop:Equality} allows us to rewrite this difference as
\begin{equation}
  \bar \Delta(z_1,\dots,z_N) = (\overline\Psi_N)_{1,n+2,\dots,N}(z_1,\dots,z_N)-  (\overline\Psi_N)_{1,n+2\dots,N}(z_1^{-1},\dots,z_N).
\end{equation}
Using the alternative integral formula \eqref{eqn:DefPsiBarCI}, one may write this difference in terms of a single contour integral similar to \eqref{eqn:DiffSimple}. Following the same lines as above, one shows that this contour integral vanishes. This proves \eqref{eqn:ReflectionLeft2} for the special choice $a_i=i+1$. The general relation \eqref{eqn:ReflectionLeft2} follows from \cref{lem:ExchangeComponents2}. 
\end{proof}

%%%%%%%%%%%%%%%%%%%%%%%%%%%%%%%%%%%%%%%%%%%%%%%%%%%%%%%%%%%%%%%%%%%%%%%%%%%%%%%%%%%%%%%%%%%%%
\subsubsection*{The boundary quantum Knizhnik-Zamolodchikov equations}
%%%%%%%%%%%%%%%%%%%%%%%%%%%%%%%%%%%%%%%%%%%%%%%%%%%%%%%%%%%%%%%%%%%%%%%%%%%%%%%%%%%%%%%%%%%%%

Let us introduce for each $i=1,\dots,N$ a so-called \textit{scattering operator}
\begin{multline} 
  \label{eqn:DefSi}
S^{(i)}(z_1,\dots,z_N)= \prod_{j=1,\dots,i-1}^\curvearrowleft \check R_{j,j+1}(s^2z_i/z_j)K_1(s^2z_i;\beta)\prod_{j=1,\dots,i-1}^\curvearrowright\check R_{j,j+1}(s^2z_iz_j)\\
  \times\prod_{j=i,\dots,N-1}^\curvearrowright\check R_{j,j+1}(s^2z_iz_{j+1}) K_N(sz_i;s\bar \beta)\prod_{j=i,\dots,N-1}^\curvearrowleft 
  \check R_{j,j+1}(z_i/z_{j+1}).
\end{multline}
\Cref{prop:Exchange,prop:Reflection} imply that the vector $|\Psi_N\rangle$ obeys the \textit{boundary quantum Knizhnik-Zamolodchikov equations}
\begin{equation}
     \label{eqn:bqKZ}
     S^{(i)}(z_1,\dots,z_N)|\Psi_N(\dots,z_i,\dots)\rangle =|\Psi_N(\dots,s^2 z_i,\dots)\rangle, \quad i=1,\dots,N.
\end{equation}
These equations appear in many integrable systems with boundaries \cite{cherednik:92,jimbo:95,stokman:15,reshetikhin:15,reshetikhin:18}. They are compatible thanks to the commutation relations 
\begin{equation}
  S^{(i)}(z_1,\dots,s^2 z_j, \dots,z_N) S^{(j)}(z_1, \dots,z_N)=S^{(j)}(z_1,\dots,s^2 z_i, \dots,z_N) S^{(i)}(z_1, \dots,z_N),
\end{equation}
for each  $1\leqslant i,j \leqslant N$. These commutation relations follow from the Yang-Baxter equation \eqref{eqn:YBE} and the boundary Yang-Baxter equation \eqref{eqn:bYBE}.

%%%%%%%%%%%%%%%%%%%%%%%%%%%%%%%%%%%%%%%%%%%%%%%%%%%%%%%%%%%%%%%%%%%%%%%%%%%%%%%%%%%%%%%%%%%%%
\subsection{Polynomiality}
\label{sec:Polynomiality}
%%%%%%%%%%%%%%%%%%%%%%%%%%%%%%%%%%%%%%%%%%%%%%%%%%%%%%%%%%%%%%%%%%%%%%%%%%%%%%%%%%%%%%%%%%%%%
In this section, we show that the components of $|\Psi_N\rangle$ are Laurent polynomials in the variables $z_1,\dots,z_N$, and determine their degrees. To this end, we examine the action of the divided difference operator on Laurent polynomials. We then apply the results of this investigation to the components.

%%%%%%%%%%%%%%%%%%%%%%%%%%%%%%%%%%%%%%%%%%%%%%%%%%%%%%%%%%%%%%%%%%%%%%%%%%%%%%%%%%%%%%%%%%%%%
\subsubsection*{The divided difference operator}
%%%%%%%%%%%%%%%%%%%%%%%%%%%%%%%%%%%%%%%%%%%%%%%%%%%%%%%%%%%%%%%%%%%%%%%%%%%%%%%%%%%%%%%%%%%%%

We consider the divided difference operator $\delta$ defined in \eqref{eqn:DefDividedDiff} acting on a Laurent polynomial $f$ in $z,w$. For a special class of Laurent polynomials, the action again yields a Laurent polynomial.
\begin{lemma}
  \label{lem:DeltaParity}
  Let $f$ be a Laurent polynomial with 
  \begin{alignat}{2}
     f(-z,w) &= \epsilon f(z,w),\qquad  f(z,- w) &&= - \epsilon f(z,w), \\
     \intertext{  where $\epsilon^2=1$, then $\delta f$ is a Laurent polynomial with the property}
     \delta f(-z,w) &= -\epsilon f(z,w), \quad \delta f(z,- w) &&=  \epsilon f(z,w).    \label{eqn:DeltaFParity} 
  \end{alignat}
\end{lemma} 
\begin{proof}
  The definition of the divided difference operator implies that $\delta f$ is a rational function that obeys \eqref{eqn:DeltaFParity}. We compute the limits
  \begin{equation}
    \lim_{z\to w} \delta f(z,w)= \frac{w}{2} \left(- \left(\frac{q^2+1}{qw}\right) f(w,w) + [q]\left(\frac{\partial f(w,w)}{\partial w}- \frac{\partial f(w,w)}{\partial z}\right)\right), 
  \end{equation}
  and, using \eqref{eqn:DeltaFParity},
  \begin{equation}
    \lim_{z\to -w} \delta f(z,w) = 
    - \epsilon  \lim_{z\to w} \delta f(z,w).
  \end{equation}
  They imply that $\delta f(z,w)$ has no poles at $z=\pm w$. Hence, it is a Laurent polynomial with respect to $z$, and similarly with respect to $w$.
\end{proof}

For each Laurent polynomial $f$ in $z,w$, there are integers $d^\pm$ and $\bar d^\pm$, with $d^-\leqslant d^+$ and $\bar d^- \leqslant \bar d^+$, such that
\begin{equation}
  f(z,w) = \sum_{k=d^-}^{d^+} c_k(w)z^k=\sum_{k=\bar d^{-}}^{\bar d^+} \bar c_k(z)w^k,
\end{equation}
with non-zero $c_{d^\pm}(w),\,\bar c_{\bar d^\pm}(z)$. We refer to $d^-$ as the lower degree and $d^+$ as the upper degree of $f$ with respect to $z$, and use the same terminology for $\bar d^-$ and $\bar d^+$ as degrees with respect to $w$. We also use the following notation:
\begin{equation}
  \deg_z^\pm f = d^\pm, \quad \deg_w^\pm f = \bar d^\pm.
\end{equation}

The degrees are not preserved by the action of the divided difference operator. They can, however, not change arbitrarily as show the next two lemmas.
\begin{lemma}
  \label{lem:Degree1}
  Let $f$ be as in \cref{lem:DeltaParity}.
  \begin{enumerate}
    \item[(i)]
    Let $m= \deg_z^+ f-1$ and suppose that $\deg_w^+ f\leqslant m$, then
  \begin{equation}
    \deg_z^+ \delta f \leqslant m \quad\text{and}\quad \deg^+_w \delta f = m+1.
  \end{equation}
  \item[(ii)] 
  Let $m= \deg_w^+ f-1$ and suppose that $\deg_z^+ f \leqslant m$, then
  \begin{equation}
    \deg^+_z \delta f = m+1 \quad\text{and}\quad \deg_w^+ \delta f \leqslant m .
  \end{equation}
 \end{enumerate}
\end{lemma}
\begin{proof}
  The proofs of \textit{(i)} and \textit{(ii)} are very similar. We only present the proof of \textit{(i)}.
  
  Let $m' = \deg_w^+f$. First, we analyse $\delta f(z,w)$ for $z\to \infty$. We find
  \begin{align}
     \delta f(z,w) &= \left(-q^{-1} \bar c_{m'}(w) z^{m'}+
     O(z^{m'-1})     \right)+\left([q] w c_{m+1}(w)z^{m} +
     O(z^{m-1}) 
     \right).     
  \end{align}
  This expression allows us to conclude that $\deg_z^+ \delta f\leqslant \max(m,m') \leqslant m$.
  Second, we consider $w\to \infty$ and obtain
  \begin{align}  
   \delta f(z,w) &= \left(-q c_{m+1}(z) w^{m+1}+O(w^m)\right)+\left(-[q] z\bar c_{m'}(z)w^{m'-1} +O(w^{m'-2})\right).
  \end{align}
  This expression implies that $\deg_w^+ \delta f = m+1$, since $q c_{m+1}(z)$ does not vanish identically and $m+1> m'-1$.
\end{proof}

\begin{lemma}
  \label{lem:Degree2}
  Let $f$ be as in \cref{lem:DeltaParity}.
 \begin{enumerate}
   \item[(i)]  Let $m=\deg_z^- f$ and suppose that $\deg_w^- f \geqslant m+1$, then
  \begin{equation}
    \deg_z^- \delta f \geqslant m+1 \quad \text{and} \quad \deg^-_w \delta f = m.
  \end{equation}
\item[(ii)] 
 Let $m=\deg_w^-f$ and suppose that $\deg_z^- f \geqslant m+1$, then
  \begin{equation}
    \deg_z^- \delta f = m \quad \text{and} \quad \deg^-_w \delta f  \geqslant m+1.
  \end{equation}
  \end{enumerate}
\end{lemma}
\begin{proof}
The proof follows from the analysis of $\delta f(z,w)$ as $z\to 0$ and $w\to 0$. It is similar to the proof of \cref{lem:Degree1}.
\end{proof}

%%%%%%%%%%%%%%%%%%%%%%%%%%%%%%%%%%%%%%%%%%%%%%%%%%%%%%%%%%%%%%%%%%%%%%%%%%%%%%%%%%%%%%%%%%%%%
\subsubsection*{Polynomiality of the vector}
%%%%%%%%%%%%%%%%%%%%%%%%%%%%%%%%%%%%%%%%%%%%%%%%%%%%%%%%%%%%%%%%%%%%%%%%%%%%%%%%%%%%%%%%%%%%%

We now show that the components of the vector $|\Psi_N\rangle$ are Laurent polynomials in the variables $z_1,\dots,z_N$ and find (bounds for) their degrees with respect to each $z_i$. We treat the cases of even and odd $N$ separately in the two following propositions. 

\begin{proposition}
\label{prop:DegreePsiEven}
Let $n\geqslant 1$, $a_1,\dots,a_n$ and $i$ be integers with $1\leqslant a_1< \dots < a_n\leqslant 2n$ and $1\leqslant i\leqslant 2n$, then we have:
\begin{enumerate}
   \item[\textit{(i)}] The component $(\Psi_{2n})_{a_1,\dots,a_n}$ is a Laurent polynomial with respect to $z_i$.
   \item[\textit{(ii)}] If $i\in \{a_1,\dots,a_n\}$ then $(\Psi_{2n})_{a_1,\dots,a_n}$ is an odd function of $z_i$ with
  \begin{equation}
    \deg_{z_i}^\pm (\Psi_{2n})_{a_1,\dots,a_n}=\pm (2n-1).
  \end{equation}
  \item[\textit{(iii)}] If $i\notin \{a_1,\dots,a_n\}$ then $(\Psi_{2n})_{a_1,\dots,a_n}$ is an even function of $z_i$ with
  \begin{equation}
    \deg_{z_i}^- (\Psi_{2n})_{a_1,\dots,a_n}\geqslant -2(n-1),\quad \deg_{z_i}^+ (\Psi_{2n})_{a_1,\dots,a_n}\leqslant 2(n-1).
    \label{eqn:DegInEq1}
  \end{equation}
\end{enumerate}
\end{proposition}
\begin{proof}
  The proof is based on recurrence. First, we note that \textit{(i)}, \textit{(ii)} and \textit{(iii)} hold for the special component $(\Psi_{2n})_{1,\dots,n}$, as readily follows from \eqref{eqn:SpecialComponent}.

 Second, we show that \textit{(i)}, \textit{(ii)} and \textit{(iii)} are preserved under the action of a divided difference operator on any component. To this end, let us consider integers $\bar a_1,\dots,\bar a_n$ with $1 \leqslant \bar a_1 < \cdots < \bar a_n \leqslant N$ such that there is $j=1,\dots,2n-1$ with $j\in \{\bar a_1,\dots, \bar a_n\}$ but $j+1\notin \{\bar a_1,\dots,\bar a_n\}$. According to \cref{lem:ExchangeComponents2}, we have
  \begin{equation}
  \label{eqn:ExchangeComponentsPsi}
  (\Psi_{2n})_{\bar a_1,\dots,j+1,\dots,\bar a_n}(\dots,z_j,z_{j+1},\dots)=\delta_j (\Psi_{2n})_{\bar a_1,\dots,j,\dots,\bar a_n}(\dots,z_j,z_{j+1},\dots).
\end{equation}
Let us now suppose that \textit{(i)}, \textit{(ii)} and \textit{(iii)} hold for the component $\Psi_{\bar a_1,\dots,j,\dots,\bar a_n}(\dots,z_j,z_{j+1},\dots)$. We apply \cref{lem:DeltaParity} to \eqref{eqn:ExchangeComponentsPsi}. It implies that $\Psi_{\bar a_1,\dots,j+1,\dots,\bar a_n}(\dots,z_j,z_{j+1},\dots)$ is an even Laurent polynomial in $z_j$, and an odd Laurent polynomial in $z_{j+1}$. Moreover, it follows from \cref{lem:Degree1}\textit{(i)} that
\begin{equation}
  \deg_{z_j}^+(\Psi_{2n})_{\bar a_1,\dots,j+1,\dots,\bar a_n}\leqslant 2(n-1), \quad \deg_{z_{j+1}}^+(\Psi_{2n})_{\bar a_1,\dots,j+1,\dots,\bar a_n}= 2n-1.
\end{equation}
Likewise, we may apply \cref{lem:Degree2}\textit{(ii)} to conclude that
\begin{equation}
  \deg_{z_j}^-(\Psi_{2n})_{\bar a_1,\dots,j+1,\dots,\bar a_n}\geqslant -2(n-1), \quad \deg_{z_{j+1}}^-(\Psi_{2n})_{\bar a_1,\dots,j+1,\dots,\bar a_n}= -(2n-1).
\end{equation}
Since all other variables remain unaffected, we conclude that \textit{(i)}, \textit{(ii)} and \textit{(iii)} hold for the component $(\Psi_{2n})_{\bar a_1,\dots,j+1,\dots,\bar a_n}(\dots,z_j,z_{j+1},\dots)$. 

Third, the statements \textit{(i)}, \textit{(ii)} and \textit{(iii)} follow for each component $(\Psi_{2n})_{a_1,\dots,a_n}$ as one can obtain it through the action of a (finite) product of divided difference operators on the component $(\Psi_{2n})_{1,\dots,n}$.
\end{proof}
\begin{proposition}
\label{prop:DegreePsiOdd} Let $n\geqslant 1$, $a_1,\dots,a_n$ and $i$ be integers with $1\leqslant a_1< \dots < a_n\leqslant 2n+1$ and $1\leqslant i\leqslant 2n+1$, then we have:
\begin{enumerate}
  \item[\textit{(i)}] The component $(\Psi_{2n+1})_{a_1,\dots,a_n}$ is a Laurent polynomial with respect to $z_i$.
  \item[\textit{(ii)}] If $i\in \{a_1,\dots,a_n\}$ then $(\Psi_{2n+1})_{a_1,\dots,a_n}$ is an odd function of $z_i$ with
  \begin{equation}
    \deg_{z_i}^- (\Psi_{2n+1})_{a_1,\dots,a_n}\geqslant -(2n-1), \quad \deg_{z_i}^+ (\Psi_{2n+1})_{a_1,\dots,a_n}\leqslant 2n-1.
        \label{eqn:DegInEq2}
     \end{equation}
  \item[\textit{(iii)}] If $i\notin \{a_1,\dots,a_n\}$ then $(\Psi_{2n+1})_{a_1,\dots,a_n}$ is an even function of $z_i$ with
  \begin{equation}
    \deg_{z_i}^\pm (\Psi_{2n+1})_{a_1,\dots,a_n}=\pm 2n.
  \end{equation}
  \end{enumerate}
\end{proposition}
\begin{proof}
  The proof is very similar to the proof of \cref{prop:DegreePsiEven}. We only mention a few minor differences. First, we note that \textit{(i)}, \textit{(ii)} and \textit{(iii)} hold for the component $(\Psi_{2n+1})_{1,\dots,n}$. Second, the argument that the action of the divided difference operators preserves \textit{(i)}, \textit{(ii)} and \textit{(iii)} follows through but uses \cref{lem:Degree1}\textit{(ii)} and \cref{lem:Degree2}\textit{(i)}. Third, the statements hold for each component $(\Psi_{2n+1})_{a_1,\dots,a_n}$ as we can obtain it    through the action of a (finite) product of divided difference operators on $(\Psi_{2n+1})_{1,\dots,n}$.
\end{proof}

\subsubsection*{Relations between even and odd size}

For each $s\in \{\uparrow,\downarrow\}$ and $i=1,\dots,N+1$, let $\Theta_i^s:V^N\to V^{N+1}$ be the linear operator whose action on the canonical basis vectors of $V^N$ is given by
\begin{equation}
  \Theta_i^s|s_1\cdots s_{i-1}s_i \cdots s_{N}\rangle=|s_1\cdots s_{i-1}\rangle\otimes |s\rangle\otimes |s_{i}\cdots s_{N}\rangle.
\end{equation}
In the next proposition, we use this operator to establish a relation between the vectors $|\Psi_N\rangle$ and $|\Psi_{N-1}\rangle$.

\begin{proposition}
\label{prop:BraidLimits}
Let $n\geqslant 1$. For each $i=1,\dots,2n$, we have
\begin{align}
\label{eqn:LimitZTo0}
&\smashoperator[l]{\lim_{z_i\to 0}} z_i^{2n-1}|\Psi_{2n}\rangle 
= (-1)^{n+i+1}\beta^{-1} q^{-\frac{3(i-1)}{2}+\frac{1}{2}{\scriptstyle \sum}_{j=1}^{i-1}\sigma_{j}^z} \Theta^\downarrow_i|\Psi_{2n-1}(z_1,\dots,z_{i-1},z_{i+1},\dots,z_{2n})\rangle,\\
  \label{eqn:LimitZTo02}
&\smashoperator[l]{\lim_{z_i\to \infty}} z_i^{-(2n-1)}|\Psi_{2n}\rangle 
= (-1)^{n+i}\beta q^{\frac{3(i-1)}{2}-\frac{1}{2} {\scriptstyle \sum}_{j=1}^{i-1}\sigma_j^z}
	\Theta^\downarrow_i|\Psi_{2n-1}(z_1,\dots,z_{i-1},z_{i+1},\dots,z_{2n})\rangle.
\end{align}
Likewise, for each $i=1,\dots,2n+1$, we have
\begin{align}  \label{eqn:LimitZTo03}
&\smashoperator[l]{\lim_{z_i\to 0}} z_i^{2n}|\Psi_{2n+1}\rangle 
= (-1)^{i-1} q^{-\frac{3(i-1)}{2}-\frac{1}{2}\sum_{j=1}^{i-1}\sigma_{j}^z}\Theta^\uparrow_i|\Psi_{2n}(z_1,\dots,z_{i-1},z_{i+1},\dots,z_{2n+1})\rangle,\\
     \label{eqn:LimitZTo04}
&\smashoperator[l]{\lim_{z_i\to \infty}} z_i^{-2n}|\Psi_{2n+1}\rangle 
= (-1)^{i-1} q^{\frac{3(i-1)}{2}+\frac{1}{2}\sum_{j=1}^{i-1} \sigma_j^z}\Theta^\uparrow_i|\Psi_{2n}(z_1,\dots,z_{i-1},z_{i+1},\dots,z_{2n+1})\rangle.
\end{align}
\end{proposition}
\begin{proof}
The proofs of the first two relations are similar. Hence, we focus on the proof of \eqref{eqn:LimitZTo0}. First, we consider the case where $i=1$. It follows from \cref{prop:DegreePsiEven} that there is a vector $|\Phi\rangle=|\Phi(z_2,\dots,z_{2n})\rangle\in V^{2n-1}$ such that
  \begin{equation}
    \lim_{z_{1}\to 0} z_{1}^{2n-1}|\Psi_{2n}(z_1,z_2\dots,z_{2n})\rangle = |{\downarrow}\rangle\otimes |\Phi(z_2,\dots,z_{2n})\rangle.
  \end{equation}
  This vector is of the form \eqref{eqn:DefPhi} (with $n$ replaced by $n-1$) and obeys the exchange relations. Moreover, it has the component
  \begin{align}
    \Phi_{1,\dots,n-1}(z_2,\dots,z_{2n})&=\lim_{z_{1}\to 0} z_1^{2n-1}  (\Psi_{2n})_{1,\dots,n}(z_1,z_2\dots,z_{2n})\\ 
    & = (-1)^n \beta^{-1} (\Psi_{2n-1})_{1,\dots,n-1}(z_2,\dots,z_{2n}).
  \end{align}
We apply \cref{lem:ExchangeComponents2} and conclude $|\Phi\rangle = (-1)^n \beta^{-1}|\Psi_{2n-1}(z_2,\dots,z_{2n})\rangle$. This ends the proof for $i=1$.
  
  Second, for $i=2,\dots, 2n$, we write
  \begin{equation}
    |\Psi_{2n}(z_1,\dots,z_i,\dots,z_{2n})\rangle= \prod_{j=1,\dots,i-1}^{\curvearrowleft}\check R_{j,j+1}(z_i/z_j)|\Psi_{2n}(z_i,z_1,\dots,z_{i-1},z_{i+1},\dots,z_{2n})\rangle. 
  \end{equation}
  Using the relation $\lim_{z\to 0}\check R_{j,j+1}(z) = -q^{-\frac32-\frac12 \sigma^z_j\sigma^z_{j+1}}\piJ_{j,j+1}$ and the result for $i=1$ leads to the relation \eqref{eqn:LimitZTo0}.

Moreover, the proofs of \eqref{eqn:LimitZTo03} and \eqref{eqn:LimitZTo04} are also similar. Therefore, we only prove \eqref{eqn:LimitZTo03}. First, we consider the case $i=2n+1$. By \cref{prop:DegreePsiEven}, there exists a vector $|\Phi\rangle = | \Phi(z_1,\dots, z_{2n})\rangle \in V^{2n}$ such that 
\begin{equation}
\lim_{z_{2n+1}\to 0}z_{2n+1}^{2n} |\Psi_{2n+1}(z_1,\dots,z_{2n+1})\rangle = |\Phi(z_1,\dots, z_{2n})\rangle \otimes |{\uparrow}\rangle.
\end{equation}
  This vector is of the form \eqref{eqn:DefPhi}, satisfies the exchange relations and has the component
  \begin{align}
    \Phi_{1,\dots,n}(z_1,\dots, z_{2n})
    &= \lim_{z_{2n+1}\to 0}z_{2n+1}^{2n}   (\Psi_{2n+1})_{1,\dots,n}(z_1,z_2\dots,z_{2n+1})\\ 
    & = q^{-3n} (\Psi_{2n})_{1,\dots,n}(z_1,\dots,z_{2n}).
  \end{align}
\cref{lem:ExchangeComponents2} implies that $|\Phi\rangle = q^{-3n} |\Psi_{2n}(z_1,\dots,z_{2n})\rangle$. A direct inspection shows that this is the result \eqref{eqn:LimitZTo03} for $i=2n+1$.

Second, for $i=1, \dots, 2n$, we have
  \begin{equation}
    |\Psi_{2n+1}(z_1,\dots,z_{2n+1})\rangle=\! \prod_{j=i,\dots,2n}\! \check R_{j,j+1}(z_{j+1}/z_i)|\Psi_{2n+1}(z_1,\dots,z_{i-1},z_{i+1},\dots,z_{2n+1}, z_i)\rangle. 
      \end{equation}
  We use the relation $\lim_{z \to \infty} \check R_{j,j+1}(z) = -q^{\frac{3}{2} + \frac 12 \sigma_j^z\sigma_{j+1}^z} \piJ_{j,j+1}$ and the result for the case $i=2n+1$ to obtain \eqref{eqn:LimitZTo03}.
\end{proof}

The preceding proposition allows us to show that the bounds in the inequalities \eqref{eqn:DegInEq1} and \eqref{eqn:DegInEq2} of \cref{prop:DegreePsiEven,prop:DegreePsiOdd} are actually equal to the degrees. 
\begin{proposition}
  \label{prop:DegreePsiImproved}
  The property (ii) of \cref{prop:DegreePsiEven} holds with \eqref{eqn:DegInEq1} replaced by
  \begin{equation}
    \label{eqn:DegPsiEvenSpinDown}
    \deg_{z_i}^\pm (\Psi_{2n})_{a_1,\dots,a_n} = \pm 2(n-1).
  \end{equation}
  Likewise, the property (iii) of \cref{prop:DegreePsiOdd} holds with \eqref{eqn:DegInEq2} replaced by
  \begin{equation}
    \deg_{z_i}^\pm (\Psi_{2n+1})_{a_1,\dots,a_n} = \pm (2n-1).
  \end{equation}
\end{proposition}
\begin{proof}
The proofs of the two properties are similar. Hence, let
us prove that \cref{prop:DegreePsiEven} holds with \eqref{eqn:DegPsiEvenSpinDown}. To this end, let $j$ and $a_1,\dots,a_n$ be integers with $1\leqslant j \leqslant 2n$ and $1\leqslant a_1 < \dots < a_\ell = j < \dots < a_n\leqslant 2n$. From \cref{prop:BraidLimits}, it follows that
   \begin{multline}
    \label{eqn:LimitComponents}
   \lim_{z_j \to 0}  z_j^{2n-1} (\Psi_{2n})_{a_1,\dots,a_{\ell-1},j,a_{\ell+1},\dots,a_n} =(-1)^{n+j+1}\beta^{-1}q^{-(\ell+j-2)} \\
\times (\Psi_{2n-1})_{a_1,\dots,a_{\ell-1},a_{\ell+1}-1,\dots,a_n-1}(z_1,\dots,z_{j-1},z_{j+1},\dots,z_{2n}).
\end{multline}
  Now, choose an integer $1\leqslant i \leqslant 2n$ with $i\notin \{a_1,\dots,a_n\}$. By \cref{prop:DegreePsiOdd}, the right-hand side of \eqref{eqn:LimitComponents} is a Laurent polynomial in $z_i$ with lower degree $-2(n-1)$ and upper degree $+2(n-1)$. According to \cref{prop:DegreePsiEven}, these are equal to the bounds on the degrees of $(\Psi_{2n})_{a_1,\dots,a_{\ell-1},j,a_{\ell+1},\dots,a_n}$ with respect to $z_i$. Since the lower degree can only increase and the upper degree can only decrease when taking the limit on the left-hand side of \eqref{eqn:LimitComponents}, we have
  \begin{equation}
    \deg_{z_i}^\pm (\Psi_{2n})_{a_1,\dots,a_{\ell-1},j,a_{\ell+1},\dots,a_n} = \pm 2(n-1).
  \end{equation}
Since this statement holds for all $j\in \{ a_1,\dots, a_n\}$, we obtain \eqref{eqn:DegPsiEvenSpinDown}.
\end{proof}

%%%%%%%%%%%%%%%%%%%%%%%%%%%%%%%%%%%%%%%%%%%%%%%%%%%%%%%%%%%%%%%%%%%%%%%%%%%%%%%%%%%%%%%%%%%%%
\subsection{Parity}
\label{sec:Parity}
%%%%%%%%%%%%%%%%%%%%%%%%%%%%%%%%%%%%%%%%%%%%%%%%%%%%%%%%%%%%%%%%%%%%%%%%%%%%%%%%%%%%%%%%%%%%%
In this section, we establish the behaviour of $|\Psi_N\rangle$ under a parity transformation. To this end, we use an explicit formula for the component $(\Psi_N)_{\bar n+1,\dots,N}$. Obtaining it from the contour-integral formulas appears to be difficult. Hence, we compute it via factor exhaustion.
\begin{proposition}
  \label{prop:SpecialComponent2}
  We have the component
  \begin{equation}
    \label{eqn:SpecialComponent2}
    (\Psi_N)_{\bar n+1,\dots,N} = (-1)^{(N+1)n}\prod_{i=\bar n+1}^N [q\beta z_i] \prod_{1\leqslant i < j \leqslant \bar n}[qz_j/z_i][qz_i z_j]\prod_{\bar n+1 \leqslant i < j \leqslant N} [qz_j/z_i][q^2z_i z_j].
  \end{equation}
\end{proposition}
\begin{proof}
\Cref{lem:ExchangeComponents1} implies the factorisation
\begin{align}
  \label{eqn:PsiComponentIntermediate}
  (\Psi_N)_{\bar n+1,\dots,N}=f_N(z_1,\dots,z_{\bar n};z_{\bar n+1},\dots,z_N)\prod_{1\leqslant i < j \leqslant \bar n}[qz_j/z_i]\prod_{\bar n+1 \leqslant i < j \leqslant N} [qz_j/z_i],
\end{align}
where $f_N(z_1,\dots,z_{\bar n};z_{\bar n+1},\dots,z_N)$ is a Laurent polynomial with respect to each of its arguments. It is separately symmetric in the variables $z_1,\dots,z_{\bar n}$, and $z_{\bar n+1},\dots,z_N$. Furthermore, the reflection relations \eqref{eqn:ReflectionRight2} and \eqref{eqn:ReflectionLeft2} lead to
\begin{align}
    [\bar \beta s^2 z_N]\prod_{i=\bar n+1}^{N-1}[qz_N/z_i]f_N(\dots;\dots,z_N) &= [\bar \beta z_{N}^{-1}]\prod_{i=\bar n+1}^{N-1}[q/(s^2z_Nz_i)]f_N(\dots;\dots,s^{-2}z_N^{-1}),\\
     \prod_{i=2}^{\bar n}[qz_j/z_1]f_N(z_1,\dots,;\dots) &= \prod_{i=2}^{\bar n}[qz_jz_1]f_N(z_1^{-1},\dots;\dots).
\end{align}
Taking into account the symmetry of $f_N$, we obtain
\begin{equation}
\label{eqn:fPsiComponent}
f_N(z_1,\dots,z_{\bar n};z_{\bar n+1},\dots,z_N)=C_N
\prod_{i=\bar n+1}^N [\bar \beta z_i^{-1}]
 \prod_{1\leqslant i < j \leqslant \bar n}[qz_iz_j]\prod_{\bar n+1 \leqslant i < j \leqslant N} [q/(s^2z_iz_j)].
\end{equation}
It follows from \cref{prop:DegreePsiEven,prop:DegreePsiOdd} that $C_N$ is a Laurent polynomial in $z_i$ with degrees $\deg_{z_i}^\pm C_N=0$ for each $i=1,\dots,N$. Therefore, $C_N$ is a constant. To find it, we consider the limits where $z_1 \to \infty$ and $z_N \to \infty$. Using \cref{prop:BraidLimits}, we obtain the recurrence relations
\begin{equation}
C_{2n+1} = (-1)^n C_{2n}, \quad C_{2n} = 
  (q^3/s^2)^{n-1}  \bar \beta \beta q\,    C_{2n-1},
\end{equation}
for each $n\geqslant 1$. The initial condition $C_1=1$ implies
\begin{align}
  \label{eqn:CResult}
  C_{2n+1} = (-1)^n C_{2n} = (-1)^{n(n+1)/2}
  (q^3/s^2)^{n(n-1)/2}(\bar \beta \beta q)^{n}.
\end{align}
  We substitute \eqref{eqn:fPsiComponent} and \eqref{eqn:CResult} into \eqref{eqn:PsiComponentIntermediate}, and simplify the resulting expression with the help of \eqref{eqn:RelationSQ} and \eqref{eqn:BetaBarBeta}. This yields \eqref{eqn:SpecialComponent2}.
\end{proof}

In the next proposition, we compute the action of the parity operator $\mathcal P$ onto $|\Psi_N\rangle$. To this end, we write $|\Psi_N\rangle = |\Psi_N(z_1,\dots,z_N;\beta)\rangle$ to stress the vector's dependence on the parameter $\beta$.
\begin{proposition}
  \label{prop:Parity}
  We have
  \begin{equation}
    \mathcal P|\Psi_N(z_1,\dots,z_N;\beta)\rangle = \epsilon_N |\Psi_N(s^{-1}z_N^{-1},\dots,s^{-1}z_1^{-1};q^2s^{-1}\beta^{-1})\rangle,
  \end{equation}
  where $\epsilon_N = (q^3s^{-2})^{(N+1)n}$.
\end{proposition}
\begin{proof}
We consider the vector
\begin{equation}
  |\Phi\rangle = \epsilon_N \mathcal P|\Psi_N(s^{-1}z_N^{-1},\dots,s^{-1}z_1^{-1};q^2s^{-1}\beta^{-1})\rangle.
\end{equation}
Using the fact that $\check R(z)$ is a symmetric matrix, it is straightforward to show that this vector obeys the exchange relations. Moreover, it has the component
\begin{equation}
  \Phi_{1,\dots,n} = \epsilon_N (\Psi_N)_{\bar n+1,\dots,N}(s^{-1}z_N^{-1},\dots,s^{-1}z_1^{-1};q^2s^{-1}\beta^{-1}).
\end{equation}
We compute the right-hand side by using the explicit formula for $(\Psi_N)_{\bar n+1,\dots,N}$ given in \cref{prop:SpecialComponent2}. This leads to $\Phi_{1,\dots,n}=(\Psi_N)_{1,\dots,n}$. Hence, by \cref{prop:ExchangeProperty} we have $|\Phi\rangle = |\Psi_N\rangle$.
\end{proof}

%%%%%%%%%%%%%%%%%%%%%%%%%%%%%%%%%%%%%%%%%%%%%%%%%%%%%%%%%%%%%%%%%%%%%%%%%%%%%%%%%%%%%%%%%%%%%
\section{The transfer matrix of the six-vertex model}
\label{sec:6V}
%%%%%%%%%%%%%%%%%%%%%%%%%%%%%%%%%%%%%%%%%%%%%%%%%%%%%%%%%%%%%%%%%%%%%%%%%%%%%%%%%%%%%%%%%%%%%
In this section, we discuss the transfer matrix of the six-vertex model on a strip. We recall its definition and some of its properties in \cref{sec:TM6V}. In \cref{sec:EVInh} we prove that the vector $|\Psi_N\rangle$ is an eigenvector of the transfer matrix if $q=\ee^{\pm 2\pi \i/3}$. Moreover, we compute the corresponding eigenvalue.

%%%%%%%%%%%%%%%%%%%%%%%%%%%%%%%%%%%%%%%%%%%%%%%%%%%%%%%%%%%%%%%%%%%%%%%%%%%%%%%%%%%%%%%%%%%%%
\subsection{The transfer matrix}
\label{sec:TM6V}
%%%%%%%%%%%%%%%%%%%%%%%%%%%%%%%%%%%%%%%%%%%%%%%%%%%%%%%%%%%%%%%%%%%%%%%%%%%%%%%%%%%%%%%%%%%%%

Following \cite{sklyanin:88}, we define the transfer matrix of the inhomogeneous six-vertex model on a strip with $N$ horizontal lines. It is an operator on the space $V^N$, given by the partial trace
\begin{equation}
  T(z|z_1,\dots,z_N) = \textup{tr}_0\left(K_0(qz;\bar \beta)\prod_{i=1,\dots,N}^\curvearrowleft R_{0,i}(z/z_i)K_0(z;\beta)\prod_{i=1,\dots,N}^\curvearrowright R_{0,i}(z z_i)\right).
\end{equation}
The operators inside the trace act on $V_0\otimes V^N = V_0\otimes V_1\otimes \cdots \otimes V_N$, where $V_0$ is an additional copy of $\mathbb C^2$ called the auxiliary space. The trace is taken  with respect to this auxiliary space.

It is well known that the Yang-Baxter equation \eqref{eqn:YBE} and the boundary Yang-Baxter equation \eqref{eqn:bYBE} imply the commutation relation
\begin{equation}
  [T(z|z_1,\dots,z_N),T(w|z_1,\dots,z_N)] = 0,
\end{equation}
for all $z,w$ and $z_1,\dots,z_N$ \cite{sklyanin:88}. The common eigenvectors of the family of commuting transfer matrices $T(z|z_1,\dots,z_N)$ parameterised by $z$ are therefore independent of $z$. In the following, we construct a common eigenvector of the family of commuting transfer matrices and compute its eigenvalue.

To this end, we use a few properties of the transfer matrices that follow from the properties of the $R$-matrix \eqref{eqn:DefR}. First, we note that the $R$-matrix is symmetric. Indeed, defining $R_{2,1}(z) =P_{1,2}R_{1,2}(z)P_{1,2}$, one checks that $R_{2,1}(z)=R_{1,2}(z)$.
Second, the $R$-matrix obeys the unitarity relation
\begin{equation}
\label{eqn:Runitarity}
R_{1,2}(z) R_{1,2}(1/z)=\bm{1},
\end{equation}
and the following relations:
\begin{align}
\label{eqn:Rcrossing1}
[q/z]\sigma_1^x \left( R_{1,2}(z)\right)^{t_1} \sigma_1^x &= -[q^2z] R_{1,2}(-1/(qz)),\\
\label{eqn:Rcrossing2}
[q/z]\sigma_1^y \left( R_{1,2}(z)\right)^{t_1} \sigma_1^y &= -[q^2z] R_{1,2}(1/(qz)),\\
\label{eqn:Rcrossing3}
\sigma_1^z  R_{1,2}(z) \sigma_1^z &=  R_{1,2}(-z).
\end{align}
Here, the superscript $t_1$ indicates the transposition with respect to the space $V_1$.

\begin{lemma} 
  \label{lem:SimpleRelations1}
  We have the relations
\begin{equation}
T(-z|z_1,\dots,z_N) = T(z|z_1,\dots,z_N),
\end{equation}
and
  \begin{equation}
    T(1/(qz)|z_1,\dots,z_L) = \frac{[\beta/z][\bar \beta/(qz)]}{[\bar \beta z][q\beta z]}\prod_{i=1}^N\left(\frac{[qz_i/ z][q/(z z_i)]}{[q^2 z/z_i][q^2 z z_i]}\right)T(z|z_1,\dots,z_L).
  \end{equation}
\end{lemma}
\begin{proof} First, we use \eqref{eqn:Rcrossing3} to write
\begin{equation}
  T(-z|z_1,\dots,z_N) = \textup{tr}_0\left(K_0(-qz;\bar \beta) \sigma_0^z \prod_{i=1,\dots,N}^\curvearrowleft R_{0,i}(z/z_i) \sigma_0^z K_0(-z;\beta) \sigma_0^z \prod_{i=1,\dots,N}^\curvearrowright R_{0,i}(z z_i)\sigma_0^z\right).\nonumber
\end{equation}
The relations $K_0(-z;\beta)=\sigma_0^zK_0(z;\beta)\sigma_0^z=K_0(z;\beta)$ and the cyclicity of the trace allow us to conclude that $T(-z|z_1,\dots,z_N) = T(z|z_1,\dots,z_N)$. 

Second, we use \eqref{eqn:Rcrossing2} to obtain
\begin{multline}
  T(1/(qz)|z_1,\dots,z_N) = \prod_{i=1}^N \frac{[qz_i/z][q/(zz_i)]}{[q^2 z/z_i][q^2 zz_i]} \\ 
 \times \textup{tr}_0\left( K_0(1/z;\bar \beta)  \sigma_0^y \prod_{i=1,\dots,N}^\curvearrowleft \left( R_{0,i}(z z_i)\right)^{t_0} \sigma_0^y K_0(1/(q z);\beta) \sigma_0^y \prod_{i=1,\dots,N}^\curvearrowright \left(R_{0,i}(z/z_i)\right)^{t_0} \sigma_0^y \right).
\end{multline}
Taking the transpose with respect to the space $V_0$, and using the cyclicity of the trace, we find
\begin{multline}
  T(1/(qz)|z_1,\dots,z_N) = \prod_{i=1}^N \frac{[qz_i/z][q/(zz_i)]}{[q^2 z/z_i][q^2 zz_i]} \\ 
 \times \textup{tr}_0\left( 
 \prod_{i=1,\dots,N}^\curvearrowright  R_{0,i}(z z_i) \sigma_0^y K_0(1/z;\bar \beta) \sigma_0^y 
 \prod_{i=1,\dots,N}^\curvearrowleft    R_{0,i}(z/z_i) \sigma_0^y  K_0(1/(qz); \beta)  \sigma_0^y \right).
\end{multline}
The $K$-matrix satisfies $\sigma^y K(1/z;\beta)\sigma^y = \frac{[\beta/z]}{[z\beta]}K(z;\beta)$ as well as the following relation
\begin{equation}
\label{eqn:KtrKRP}
   \text{tr}_{\bar 0}\left(K_{\bar 0}(qz;\bar \beta)R_{\bar 0,0}(z^2)\piJ_{\bar 0,0}\right) = \frac{[q^2z^2][\bar\beta /z]}{[z^{2}/q][qz/{\bar \beta}]}K_0(z;\bar \beta),
 \end{equation}
where we introduced a new auxiliary space $V_{\bar 0}=\mathbb C^2$. The operators inside the  trace act on $V_{\bar 0}\otimes V_0\otimes V^N$. We use these properties of the $K$-matrix to write
\begin{multline}
  T(1/(qz)|z_1,\dots,z_N) = \prod_{i=1}^N \frac{[qz_i/z][q/(zz_i)]}{[q^2 z/z_i][q^2 zz_i]} \frac{[q z/\bar \beta ][\beta/(qz)]}{[z \bar \beta][\beta q z]} \frac{[z^2/q]}{[q^2 z^2]} \\ 
 \times \textup{tr}_{\bar 0}\left( \textup{tr}_0\left( K_{\bar 0}(qz;\bar \beta)
 \prod_{i=1,\dots,N}^\curvearrowright  R_{0,i}(z z_i) R_{\bar 0,0}(z^2)
 \prod_{i=1,\dots,N}^\curvearrowleft    R_{\bar 0,i}(z/z_i) \piJ_{\bar 0,0} K_0(qz; \beta)   \right)\right).
\end{multline}
We rearrange the products of $R$-matrix inside the trace by using the Yang-Baxter equation in the form
\begin{equation}
R_{0,i}(zz_{i})R_{\bar 0,0}(z^2)R_{\bar 0,i}(z/z_i)=  R_{\bar 0,i}(z/z_i)R_{\bar 0,0}(z^2)R_{0,i}(zz_{i}).
\end{equation}
After rearrangement, we note that the permutation operator $\piJ_{\bar 0,0}$ allows us to take the products out of the trace over the space $V_0$. We obtain
\begin{multline}
  T(1/(qz)|z_1,\dots,z_N) = \prod_{i=1}^N \frac{[qz_i/z][q/(zz_i)]}{[q^2 z/z_i][q^2 zz_i]} \frac{[q z/\bar \beta ][\beta/(qz)]}{[z \bar \beta][\beta q z]} \frac{[z^2/q]}{[q^2 z^2]} \\ 
 \times \textup{tr}_{\bar 0} \left( 
  K_{\bar 0}(q z ;\bar \beta) \prod_{i=1,\dots,N}^\curvearrowleft R_{\bar 0,i}(z/z_i)
  \textup{tr}_0\left(R_{\bar 0,0}(z^2)\piJ_{\bar 0,0} K_0(qz; \beta)   \right) 
 \prod_{i=1,\dots,N}^\curvearrowright  R_{\bar 0,i}(z z_i) \right).
\end{multline}
We again use the relation \eqref{eqn:KtrKRP} and recognise the transfer matrix with auxiliary space index $\bar 0$. This concludes the proof. 
\end{proof}

\begin{lemma}
  \label{lem:SimpleRelations2}
  Let $z$ be a solution of $z^4=1$ then
  \begin{equation}
    \label{eqn:Tz41}
   T(z|z_1,\dots,z_N) = \frac{[q^2][\bar \beta/z]}{[q][\bar\beta/(q z)]} \bm{1}.
  \end{equation}
\end{lemma}
\begin{proof}
  We separately consider the cases where $z=\pm 1$ and $z=\pm \i$.
  First, if $z=\pm 1$ then we have $K_0(z;\beta) = \bm{1}$ and $R_{0,i}(z/z_i) = R_{0,i}(1/(z z_i))$. These relations allow us to write
    \begin{equation}
  T(z|z_1,\dots,z_N) = \textup{tr}_0\left(K_0(qz;\bar \beta)\prod_{i=1,\dots,N}^\curvearrowleft R_{0,i}(1/(zz_i) )\prod_{i=1,\dots,N}^\curvearrowright R_{0,i}(zz_i)\right).
  \end{equation}
   Moreover, using the unitarity relation \eqref{eqn:Runitarity}, we find
   \begin{equation}
   \label{eqn:Tzpm}
 T(z|z_1,\dots,z_N) = \textup{tr}_0\left(K_0(qz;\bar \beta)\right)=\frac{[q^2][\bar \beta]}{[q][\bar\beta/q]}\bm{1}.
  \end{equation}
  This proves \eqref{eqn:Tz41} for $z=\pm 1$.

  Second, if $z=\pm \i$ then we have $K_0(z;\beta) = \sigma_0^z,$ and $R_{0,i}(z/z_i)\sigma_0^z = \sigma_0^z R_{0,i}(1/(zz_i))$. These relations allow us to write
   \begin{equation}
  T(z|z_1,\dots,z_N) = \textup{tr}_0\left(K_0(qz;\bar \beta)\sigma_0^z \prod_{i=1,\dots,N}^\curvearrowleft R_{0,i}(1/(zz_i))\prod_{i=1,\dots,N}^\curvearrowright R_{0,i}(z z_i)\right).
  \end{equation}
   Using the unitarity relation \eqref{eqn:Runitarity}, we find
     \begin{equation}
  T(z|z_1,\dots,z_N) = \textup{tr}_0\left(K_0(qz;\bar \beta)\sigma_0^z\right)=
\frac{[q^2][\bar \beta/ z]}{[q][\bar\beta/(qz)]} \bm{1}.
 \end{equation}
    This proves \eqref{eqn:Tz41} for $z=\pm \i$.
\end{proof}

%%%%%%%%%%%%%%%%%%%%%%%%%%%%%%%%%%%%%%%%%%%%%%%%%%%%%%%%%%%%%%%%%%%%%%%%%%%%%%%%%%%%%%%%%%%%%
\subsection{The eigenvector}
\label{sec:EVInh}
%%%%%%%%%%%%%%%%%%%%%%%%%%%%%%%%%%%%%%%%%%%%%%%%%%%%%%%%%%%%%%%%%%%%%%%%%%%%%%%%%%%%%%%%%%%%%

In this section, we establish a relation between the transfer matrix of the inhomogeneous six-vertex model on the strip and the scattering operators \eqref{eqn:DefSi}. (We refer to \cite{vlaar:15} for a general discussion on the relation between transfer matrices and scattering operators.) We use this relation to show that if $q=\ee^{\pm 2\pi \i/3}$ then the vector $|\Psi_N\rangle$ is an eigenvector of the transfer matrix. Moreover, we explicitly compute the corresponding eigenvalue.

\begin{proposition}
  \label{lem:SpecialisationT}
  If $q=\ee^{\pm 2\pi \i/3}$ and \eqref{eqn:BetaBarBeta} holds then
  \begin{equation}
  \label{eqn:TziS}
  T(z_i|z_1,\dots,z_N) = -\frac{[q \beta z_i]}{[q^2\beta z_i]}S^{(i)}(z_1,\dots,z_N),
  \end{equation}
  for each $i=1,\dots,N$. Here, $S^{(i)}(z_1,\dots,z_N)$ is the operator defined in \eqref{eqn:DefSi} with $s=1$.
\end{proposition}
\begin{proof}First, we assume that $q,\beta,\bar \beta$ are generic.
  We use $R_{0,i}(1) = \piJ_{0,i}$, the symmetry of the $R$-matrix, and the properties of partial traces to write
 \begin{multline}
   T(z_i|z_1,\dots,z_N) = \prod_{j=1,\dots,i-1}^\curvearrowleft R_{i,j}(z_i/z_j)K_i(z_i;\beta) \prod_{j=1,\dots,i-1}^\curvearrowright R_{i,j}(z_iz_j) \\
   \times  \text{tr}_0\left(K_0(qz_i;\bar \beta)\prod_{j=i+1,\dots,N}^\curvearrowleft R_{0,j}(z_i/z_j)R_{0,i}(z_i^2)\prod_{j=i+1,\dots,N}^\curvearrowright  R_{i,j}(z_iz_{j})\piJ_{0,i} \right).
 \end{multline}
We rearrange the products of $R$-matrix using the Yang-Baxter equation as in the proof of \cref{lem:SimpleRelations1} and obtain
\begin{multline}    T(z_i|z_1,\dots,z_N) = \prod_{j=1,\dots,i-1}^\curvearrowleft R_{i,j}(z_i/z_j)K_i(z_i;\beta) \prod_{j=1,\dots,i-1}^\curvearrowright R_{i,j}(z_iz_j) \\
 \times  \prod_{j=i+1,\dots,N}^\curvearrowright R_{i,j}(z_iz_{j})\, \text{tr}_0\left(K_0(qz_i;\bar \beta)R_{0,i}(z_i^2)\piJ_{0,i}\right)\prod_{j=i+1,\dots,N}^\curvearrowleft R_{i,j}(z_i/z_j).
 \end{multline}
 The remaining trace is given by
 \begin{equation}
   \text{tr}_0\left(K_0(qz_i;\bar \beta)R_{0,i}(z_i^2)\piJ_{0,i}\right) = \frac{[q^2z_i^2][\bar\beta /z_i]}{[z_i^{2}/q][qz_i/{\bar \beta}]}K_i(z_i;\bar \beta).
 \end{equation}
 We rewrite the products of $R$-matrices as products of $\check R$-matrices and obtain
\begin{multline}
   T(z_i|z_1,\dots,z_N) =\frac{[q^2z_i^2][\bar\beta /z_i]}{[z_i^{2}/q][qz_i/\bar \beta]}\prod_{j=1,\dots,i-1}^\curvearrowleft \check R_{j,j+1}(z_i/z_j)\,K_1(z_i;\beta) \prod_{j=1,\dots,i-1}^\curvearrowright \check R_{j,j+1}(z_iz_j) \label{eqn:TziFinal}\\
\times  \prod_{j=i,\dots,N-1}^\curvearrowright \check R_{j,j+1}(z_iz_{j+1})\, K_N(z_i;\bar \beta)\prod_{j=i,\dots,N-1}^\curvearrowleft \check R_{j,j+1}(z_i/z_{j+1}). 
\end{multline}

Second, we assume $q=\ee^{\pm 2\pi \i/3}$ and that \eqref{eqn:BetaBarBeta} holds. In \eqref{eqn:TziFinal}, these specialisations lead to the pre-factor
\begin{equation}
   \frac{[q^2z_i^2][\bar\beta /z_i]}{[z_i^{2}/q][qz_i/\bar \beta]} = 
   - \frac{[q\beta z_i]}{[q^2\beta z_i]}.
\end{equation}
The remaining products of $\check R$- and $K$-matrices yield the operator $S^{(i)}(z_1,\dots,z_N)$ defined in \eqref{eqn:DefSi} with $s=1$.
\end{proof}

\begin{proposition}
  \label{prop:Eigenvector}
Let $q = \ee^{\pm 2\pi \i/3}$ and suppose that \eqref{eqn:BetaBarBeta} holds then we have
  \begin{equation}
    T(z|z_1,\dots,z_N)|\Psi_N\rangle=\Lambda_N(z)|\Psi_N\rangle,
    \label{eqn:TMEigenvalue1}
  \end{equation}
  where the eigenvalue is
  \begin{equation}
    \Lambda_N(z) = -\frac{[q\beta z]}{[q^2\beta z]}.
        \label{eqn:TMEigenvalue2}
  \end{equation}
\end{proposition}
\begin{proof}
We define the operator
\begin{equation}
  \bar T(z|z_1,\dots,z_N) = -[\beta/z][q^2\beta z]\Bigr(\prod_{j=1}^N[qz_j/z][q/(z_j z)]\Bigr)\,T(z|z_1,\dots,z_N).
\end{equation}
We prove that if $q = \ee^{\pm 2\pi \i/3}$ and \eqref{eqn:BetaBarBeta} holds then
\begin{equation}
  \label{eqn:TBarEV}
  \bar T(z|z_1,\dots,z_N)|\Psi_N\rangle=[\beta/z][q \beta z]\Bigr(\prod_{j=1}^N[qz_j/z][q/(z_j z)]\Bigr)\,|\Psi_N\rangle,
\end{equation}
which is equivalent to the statement of the theorem.
To this end, we note that the pre-factor on the right-hand side of \eqref{eqn:TBarEV} is a Laurent polynomial in $z$ with lower degree $-2(N+1)$ and upper degree $2(N+1)$. Likewise, the matrix elements of $\bar T(z|z_1,\dots,z_N)$ are Laurent polynomials in $z$ with lower degree at least $-2(N+1)$ and upper degree at most $2(N+1)$. Therefore, it is sufficient to show that \eqref{eqn:TBarEV} holds for at least $4N+5$ distinct values of $z$.

First, it follows from \cref{lem:SpecialisationT} that
\begin{equation}
  \bar T(z_i|z_1,\dots,z_N) = [\beta/z_i][q \beta z_i]\Bigr(\prod_{j=1}^L[qz_j/z_i][q/(z_j z_i)]\Bigr) S^{(i)}(z_1,\dots,z_N),
\end{equation}
where $S^{(i)}(z_1,\dots,z_N)$ is the operator defined in \eqref{eqn:DefSi} with $s=1$. It follows from this equality and from the boundary quantum Knizhnik-Zamolodchikov equations \eqref{eqn:bqKZ} that \eqref{eqn:TBarEV} holds if $z=z_i$ for each $i=1,\dots,N$. Moreover, \cref{lem:SimpleRelations1} allows us to conclude that it holds if $z=-z_i,1/(qz_i),-1/(q z_i)$ for each $i=1,\dots,N$, too.

Second, according to \cref{lem:SimpleRelations2}, for any solution $z$ of $z^4=1$ we have
\begin{equation}
  \bar T(z|z_1,\dots,z_L) = \Bigr(\prod_{j=1}^N[qz_j/z][q/(z_j z)]\Bigr)[\beta/z][q\beta z]\bm{1}.
\end{equation}
Together with \cref{lem:SimpleRelations1}, this equality implies that \eqref{eqn:TBarEV} trivially holds for the values $z=\pm 1, \pm \i,\pm q^{-1},\pm \i q^{-1}$.

In summary, the relation  \eqref{eqn:TBarEV} holds for $4N+8>4N+5$ distinct values of $z$ and, hence, for all $z$.
\end{proof}

Two remarks about the eigenvalue $\Lambda_N(z)$ are in order. First,
 we note that it can be written as the trace
\begin{equation}
  \label{eqn:LambdaTrace}
  \Lambda_N(z) = \text{tr}_0\left(K_0(qz;\bar\beta)K_0(z;\beta)\right),
\end{equation}
where $\bar \beta$ is a solution of \eqref{eqn:BetaBarBeta}. Formally, the right-hand side of this equality is the transfer matrix of the six-vertex model on a strip with $N=0$ vertical lines. Second, for $q=\beta = \ee^{2\pi\i/3}$ and $z_1=\dots=z_N=1$, the eigenvalue follows from the trigonometric limit of the eigenvalue problem for the supersymmetric eight-vertex model on a strip studied in the article \cite{hagendorf:20}. In that article, we also conjectured a generalisation of the eigenvalue to the inhomogeneous eight-vertex model. The expression of $\Lambda_N(z)$ found here above, once evaluated at $q=\beta=\ee^{2\pi\i/3}$ but with arbitrary $z_1,\dots,z_N$, proves the trigonometric limit of the conjectured expression.

%%%%%%%%%%%%%%%%%%%%%%%%%%%%%%%%%%%%%%%%%%%%%%%%%%%%%%%%%%%%%%%%%%%%%%%%%%%%%%%%%%%%%%%%%%%%%
\section{The homogeneous limit}
\label{sec:HomLimit}
%%%%%%%%%%%%%%%%%%%%%%%%%%%%%%%%%%%%%%%%%%%%%%%%%%%%%%%%%%%%%%%%%%%%%%%%%%%%%%%%%%%%%%%%%%%%%
In this section, we investigate the homogeneous limit of the vector $|\Psi_N\rangle$. It is convenient to define the rescaled version
\begin{equation}
  \label{eqn:DefPsiH}
  |\psi_N\rangle = (-1)^{\bar n(\bar n-1)/2} [\beta]^{-n}[q]^{-n(n-1)-\bar n(\bar n-1)}|\Psi_N(1,\dots,1)\rangle.
\end{equation}
It depends on the two parameters $q$ and $\beta$. By \cref{prop:SpecialComponent}, the vector is non-vanishing for generic $q$. Throughout this section, we use for this homogeneous limit the same notation as for the spin-chain ground-state vector of \cref{sec:SpinChainGS}. In \cref{sec:EV}, we show that if $q=\ee^{\pm 2\pi\i/3}$ then it is indeed the ground-state vector of the XXZ Hamiltonian \eqref{eqn:XXZ} with the parameters \eqref{eqn:CombinatorialPoint}, where
\begin{equation}
  \label{eqn:DefXBeta}
   x = -[q\beta]/[\beta].
\end{equation}
In \cref{sec:Components}, we show that the vector's components are given in terms of multiple contour integrals. They allow us to characterise the components as polynomials in $x$. The purpose of \cref{sec:ScalarProducts} is to find determinant formulas for a family of overlaps involving the vector. In \cref{sec:TSASM}, we formulate a conjecture between the sum of the components of $|\psi_N\rangle$ and a weighted enumeration of totally-symmetric alternating sign matrices.

%%%%%%%%%%%%%%%%%%%%%%%%%%%%%%%%%%%%%%%%%%%%%%%%%%%%%%%%%%%%%%%%%%%%%%%%%%%%%%%%%%%%%%%%%%%%%
\subsection{The ground-state eigenvalue of the XXZ Hamiltonian}
\label{sec:EV}
%%%%%%%%%%%%%%%%%%%%%%%%%%%%%%%%%%%%%%%%%%%%%%%%%%%%%%%%%%%%%%%%%%%%%%%%%%%%%%%%%%%%%%%%%%%%%

In this section, we prove \cref{thm:MainTheorem1}. To this end, we provide two auxiliary results about the XXZ Hamiltonian \eqref{eqn:XXZ} in the following two lemmas. The first lemma relates this Hamiltonian to the transfer matrix of a of the homogeneous six-vertex model on the strip 
\begin{equation}
  t(z) = T(z|1,\dots,1).
\end{equation}
\begin{lemma}
  The logarithmic derivative of $t(z)$ at $z=1$ is
  \begin{equation}
  \label{eqn:TH}
  t(1)^{-1}t'(1) = -\frac{4}{[q]}\left(H-C\bm{1}\right),
\end{equation}
where $H$ is the Hamiltonian \eqref{eqn:XXZ}, whose parameters and the constant $C$ are given by
\begin{equation}
  \label{eqn:ParamsQ}
  \Delta = \frac{[q^2]}{2[q]}, \quad p = \frac{[q][\beta^2]}{4[\beta]^2}, \quad \bar p = \frac{[q][\bar\beta^2]}{4[\bar \beta]^2}, \quad   C = \frac{3N[q^2]}{4[q]} +\frac{[q][\beta^2]}{4[\beta]^2}-\frac{[q]^2[\bar\beta^2]}{2[q^2][\bar\beta][q/\bar\beta]}.
\end{equation}
\end{lemma}
\begin{proof}
  The statement follows from a standard calculation \cite{sklyanin:88}.
\end{proof}

For the second lemma, we denote by $H_\mu$ the restriction of the Hamiltonian $H$ to the sector of magnetisation $\mu = (\bar n-n)/2$. The lemma addresses the degeneracy of its ground-state eigenvalue. 
\begin{lemma}
  \label{lem:Nondegeneracy}
 For $x>0$, the ground-state eigenvalue of $H_\mu$ is non-degenerate.
\end{lemma}
\begin{proof}
  Let $\lambda=N-1+x+x^{-1}$. One checks that, for $x>0$, the matrix $\lambda\bm 1 -H_\mu$ is a non-negative and irreducible matrix, following the arguments of \cite{yang:66} (see also \cite{hagendorf:20}). By the Perron-Frobenius theorem for non-negative matrices \cite{meyer:00}, the largest eigenvalue of $\lambda \bm 1-H_\mu$ is non-degenerate. Hence, the ground-state eigenvalue of $H_\mu$ is non-degenerate.
\end{proof}

\begin{proof}[Proof of \cref{thm:MainTheorem1}]
We divide the proof into two parts. In part 1, we show that the XXZ Hamiltonian \eqref{eqn:XXZ} with the parameters \eqref{eqn:CombinatorialPoint} possesses the eigenvalue $E_0$ given in \eqref{eqn:SpecialEV}. In part 2, we show that if $x>0$ then $E_0$ is the non-degenerate ground-state eigenvalue of the restriction $H_\mu$.

\medskip

\noindent \textit{Part 1: Existence of the eigenvalue.} Let $q=\ee^{\pm 2\pi\i/3}$ and suppose $\bar \beta$ and $\beta$ obey the relation \eqref{eqn:BetaBarBeta}. In this case,  \eqref{eqn:ParamsQ} implies that the spin chain's parameters are
\begin{equation}
  \Delta = -\frac{1}{2}, \quad p = \frac{1}{2}\left(\frac{1}{2}+\frac{[q \beta]}{[\beta]}\right), \quad \bar p = \frac{1}{2}\left(\frac{1}{2}+\frac{[\beta]}{[q \beta]}\right).
\end{equation}
They coincide with the parameters \eqref{eqn:CombinatorialPoint}, provided that $x$ is identified in terms of $\beta$ as in \eqref{eqn:DefXBeta}. Moreover, by \cref{prop:Eigenvector} we have $t(z)|\psi_N\rangle=\Lambda_N(z)|\psi_N\rangle$. We evaluate the derivative with respect to the spectral parameter $z$ at $z=1$ on both sides of this equality and use \eqref{eqn:TH}, which yields
\begin{equation}
H|\psi_N\rangle = E_0|\psi_N\rangle \quad \text{with}\quad E_0 =C -  \frac{[q]}{4}\Lambda_N(1)^{-1}\Lambda'_N(1),
\end{equation}
where $C$ is defined in \eqref{eqn:ParamsQ}. Using the expression \eqref{eqn:TMEigenvalue2} for the eigenvalue $\Lambda_N(z)$ as well as the relation \eqref{eqn:DefXBeta}, we obtain
\begin{equation}
  E_0 = -\frac{3N-1}{4}-\frac{(1-x)^2}{2x}.
\end{equation}

\medskip

\noindent \textit{Part 2: Ground-state eigenvalue.} 
We explicitly write out the dependences of $x$: $E_0=E_0(x)$, $H_\mu=H_\mu(x)$ etc. First, we note that $H_\mu(x)$ depends continuously on $x$ for $x>0$. Hence, its eigenvalues are continuous functions of $x$ for $x>0$. Second, we have shown in \cite{hagendorf:17} that $E_0(1)$ is the non-degenerate ground-state eigenvalue of $H_\mu(1)$. Third, let us suppose by contradiction that there is $x''>0$ such that $E_0(x'')$ is not the ground-state eigenvalue of $H_\mu(x'')$. Without loss of generality, we may assume that $x''>1$. (The proof for $0<x''<1$ follows the same line of arguments.) By the continuity of the eigenvalues with respect to $x$, there is $x'$ with $1< x' < x''$ such that $E_0(x)$ is the ground-state eigenvalue of $H_\mu(x)$ for $1\leqslant x \leqslant x'$ and coincides with another eigenvalue of the Hamiltonian for $x=x'$. This implies that the ground-state eigenvalue is at least doubly degenerate at $x=x'$, which contradicts \cref{lem:Nondegeneracy}.
\end{proof}

%%%%%%%%%%%%%%%%%%%%%%%%%%%%%%%%%%%%%%%%%%%%%%%%%%%%%%%%%%%%%%%%%%%%%%%%%%%%%%%%%%%%%%%%%%%%%
\subsection{Components}
\label{sec:Components}
%%%%%%%%%%%%%%%%%%%%%%%%%%%%%%%%%%%%%%%%%%%%%%%%%%%%%%%%%%%%%%%%%%%%%%%%%%%%%%%%%%%%%%%%%%%%%
In this section, we prove \cref{thm:MainTheorem2,thm:MainTheorem3}. To this end, we investigate the components of the vector $|\psi_N\rangle$. For $N\geqslant 2$, we may write
\begin{equation}
  |\psi_N\rangle = \sum_{1\leqslant a_1 < \cdots < a_n \leqslant N} (\psi_N)_{a_1,\dots,a_n}|{\uparrow}\cdots \uparrow\underset{a_1}{\downarrow}\uparrow \quad \cdots \quad\uparrow
  \underset{a_n}{\downarrow}\uparrow \cdots \uparrow\rangle.
  \label{eqn:PsiHExpansion1}
\end{equation}
Likewise, we have
\begin{equation}
  |\psi_N\rangle = \sum_{1\leqslant b_1 < \cdots < b_{\bar n} \leqslant N} (\overline \psi_N)_{b_1,\dots,b_{\bar n}}|{\downarrow}\cdots \downarrow\underset{b_1}{\uparrow}\downarrow \quad \cdots \quad\downarrow
  \underset{b_{\bar n}}{\uparrow}\downarrow \cdots \downarrow\rangle.
  \label{eqn:PsiHExpansion2}
\end{equation}
We note that \eqref{eqn:DefPsiH} fixes the following component:
\begin{equation}
  (\psi_N)_{1,\dots,n}=(\overline \psi_N)_{n+1,\dots,N} = \tau^{\bar n(\bar n-1)/2}.
\end{equation}
Here, and in the following, $\tau = -q-q^{-1}$. We note that $\tau = 1$ for $q=\ee^{\pm 2\pi\i/3}$.

%%%%%%%%%%%%%%%%%%%%%%%%%%%%%%%%%%%%%%%%%%%%%%%%%%%%%%%%%%%%%%%%%%%%%%%%%%%%%%%%%%%%%%%%%%%%%
\subsubsection*{Contour integral formulas}
%%%%%%%%%%%%%%%%%%%%%%%%%%%%%%%%%%%%%%%%%%%%%%%%%%%%%%%%%%%%%%%%%%%%%%%%%%%%%%%%%%%%%%%%%%%%%
There are several contour-integral formulas for the components of $|\psi_N\rangle$. The first type of formulas follows from the evaluation of \eqref{eqn:DefPsiCI} with $z_1=\dots=z_N=1$. Changing the integration variables to $u_i = [w_i]/[w_i/q]$ leads to the following expression:
\begin{multline}   \label{eqn:PsiHComp1} 
\hspace*{-13pt} (\psi_N)_{a_1,\dots,a_n}=  \tau^{N(N-1)/2}\!\oint \dotsi \oint\! \prod_{k=1}^n \frac{\diff u_k}{2\pi \i}\frac{(1+xu_k)(1+\tau u_k)(\tau+(\tau^2-2)u_k)(1+\tau u_k + u^2_k)^{N-2n}}{u_k^{a_k}(\tau+(\tau^2-1)u_k)^N} \\
\times \prod_{1\leqslant i < j \leqslant n}(u_j-u_i)(1+\tau(u_i+u_j)+(\tau^2-1)u_iu_j)(1+\tau u_j+u_iu_j)\\
 \times (\tau + (\tau^2-1)(u_i+u_j) + \tau (\tau^2-2)u_iu_j).
\end{multline}
The integration contour of each $u_i$ goes around $0$, but not around $-\tau/(\tau^2-1)$. Likewise, the evaluation of \eqref{eqn:DefPsiBarCI} with $z_1=\dots=z_N=1$ and a change of the integration variables to $u_i = [w_{\bar n+1-i}]/[q w_{\bar n+1-i}]$ leads to a second contour integral formula
\begin{multline}   \label{eqn:PsiBarHComp} 
  (\overline\psi_N)_{b_1,\dots,b_{\bar n}}=  \oint  \dots \oint \prod_{k=1}^{\bar n}
    \frac{\diff u_k}{2\pi \i}\frac{(1+\tau u_k + u^2_k)^{N+1-2\bar n}}{u_k^{N+1-b_{\bar n +1 -k}}(1+(x-\tau)u_k)}\\
\times\prod_{1\leqslant i \leqslant j \leqslant \bar n}(1-u_i u_j)\prod_{1\leqslant i <j \leqslant \bar n}(u_j-u_i)(1+\tau u_j + u_i u_j)(\tau+u_i+u_j).
\end{multline}
Here, the integration contour of $u_i$ goes around $0$. We also use the following third contour-integral representation of the components:
\begin{proposition}
\label{prop:PsiHComp2}
For each increasing sequence $1\leqslant a_1 < \dots < a_n \leqslant N$, we have
\begin{multline}
  \label{eqn:PsiHComp2} 
  (\psi_N)_{a_1,\dots,a_n}= \oint  \dots \oint \prod_{k=1}^n \frac{\diff u_k}{2\pi \i}\frac{(u_k+x)(1+\tau u_k+u_k^2)^{N-2n}}{u_k^{N+1-a_{n+1-k}}}\\
  \times\prod_{1\leqslant i \leqslant j \leqslant n}(1-u_i u_j)\prod_{1\leqslant i <j \leqslant n}(u_j-u_i)(1+\tau u_j + u_i u_j)(\tau+u_i+u_j).
\end{multline}
The integration contour of each $u_i$ goes around $0$.
\end{proposition}
\begin{proof}
  It follows from \cref{prop:Parity} that
  \begin{equation}
    (\Psi_N)_{a_1,\dots,a_n}= \epsilon_N  (\Psi_N)_{N+1-a_n,\dots,N+1-a_1}(s^{-1}z_N^{-1},\dots,s^{-1}z_1^{-1};q^2s^{-1}\beta^{-1}),
  \end{equation}
  where $s$ obeys \eqref{eqn:RelationSQ}. Using the integral formula \eqref{eqn:DefPsiCI} on the right-hand side leads, after some algebra, to
\begin{multline}
     (\Psi_N)_{a_1,\dots,a_n}=[q]^n \prod_{1\leqslant i < j \leqslant N}[q z_j/z_i][q z_i z_j]\oint \cdots \oint \prod_{k=1}^n \frac{\diff w_k}{\i \pi w_k} [q \beta w_k]\\
     \times \frac{\prod_{1\leqslant i < j \leqslant n} [q w_j/w_i][w_i/w_j]
     [q^2 w_i w_j]\prod_{1\leqslant i \leqslant j \leqslant n}[q w_i w_j]}{\prod_{i=1}^n\left( \prod_{j=1}^{a_i}[qw_i/z_j]\prod_{j=a_i}^N [w_i/z_j] 
     \prod_{j=1}^N
     [q w_i z_j]\right)}.
       \end{multline}
  The integration contour of each $w_i$ is a collection of positively-oriented curves around $z_j$, but not around $0,-z_j,\pm q^{-1} z_j,\,\pm q^{-1} z_j^{-1}$, $j=1,\dots,N$.
  
  We now set $z_1=\dots=z_N=1$ and change the integration variables to $u_i = [w_{n+1-i}]/[q w_{n+1-i}]$. Using the definitions \eqref{eqn:DefPsiH} and \eqref{eqn:DefXBeta}, we obtain \eqref{eqn:PsiHComp2}.
\end{proof}

For $x=0$, this proposition shows that $|\psi_N\rangle$ is the vector studied in \cite{degier:09}. We now show that this vector can also be computed from $x=\tau$, using the spin-reversal operator \eqref{eqn:SpinReversal}.
\begin{proposition}
  \label{prop:X0XTau}
  We have $|\psi_N(0)\rangle = \mathcal R|\psi_{N-1}(\tau)\rangle\otimes|{\uparrow}\rangle$.
\end{proposition}
\begin{proof}
  For $x=0$, the contour-integral expression \eqref{eqn:PsiHComp2} becomes
\begin{multline}
  (\psi_N)_{a_1,\dots,a_n}= \oint  \dots \oint \prod_{k=1}^n \frac{\diff u_k}{2\pi \i}\frac{(1+\tau u_k+u_k^2)^{N-2n}}{u_k^{N-a_{n+1-k}}}\\
  \times\prod_{1\leqslant i \leqslant j \leqslant n}(1-u_i u_j)\prod_{1\leqslant i <j \leqslant n}(u_j-u_i)(1+\tau u_j + u_i u_j)(\tau+u_i+u_j).
\end{multline}
It implies that the component vanishes for $a_n=N$ because the integrand has no pole at $u_1=0$. For $a_n < N$, we find by comparison with \eqref{eqn:PsiBarHComp} the relation
\begin{equation}
   (\psi_N(0))_{a_1,\dots,a_n} = (\overline \psi_{N-1}(\tau))_{a_1,\dots,a_n}.
\end{equation}
The proposition follows from this relation, together with \eqref{eqn:PsiHExpansion1} and \eqref{eqn:PsiHExpansion2}.
\end{proof}

%%%%%%%%%%%%%%%%%%%%%%%%%%%%%%%%%%%%%%%%%%%%%%%%%%%%%%%%%%%%%%%%%%%%%%%%%%%%%%%%%%%%%%%%%%%%%
\subsubsection*{Polynomiality}
%%%%%%%%%%%%%%%%%%%%%%%%%%%%%%%%%%%%%%%%%%%%%%%%%%%%%%%%%%%%%%%%%%%%%%%%%%%%%%%%%%%%%%%%%%%%%

The contour integral representations allow us to explicitly compute the components of $|\psi_N\rangle$ for small $N$. For example, for $N=5$ sites, the components are given by the polynomials
\begin{align}
(\psi_5)_{1,2} & = \tau^3,
  &(\psi_5)_{1,3} &= \tau^2(2+\tau^2) + x\tau^3,\nonumber \\
  (\psi_5)_{1,4} &= \tau(2+\tau^2) + x\tau^2(2+\tau^2),  &(\psi_5)_{1,5} &= x\tau(2+\tau^2),\nonumber \\
  (\psi_5)_{2,3} &= \tau(1+\tau^2) + 2x\tau^2+x^2\tau^3, &(\psi_5)_{2,4} &= 1+2\tau^2 + x\tau(3+2\tau^2)+x^2\tau^2(2+\tau^2),\\
  (\psi_5)_{2,5} &= x(1+2\tau^2)+x^2\tau(2+\tau^2),  &(\psi_5)_{3,4} &= \tau + x(1+\tau^2) + x^2\tau(1+\tau^2),\nonumber \\
  (\psi_5)_{3,5} &= x\tau + x^2(1+2\tau^2) ,  &(\psi_5)_{4,5} &= x^2\tau.\nonumber 
\end{align}
For $\tau = 1$, we recover the XXZ ground-state components \eqref{eqn:ExN5} obtained from the exact diagonalisation of spin-chain Hamiltonian. We now establish the polynomiality in $x$ of the ground-state components for arbitrary $N$.

\begin{proof}[Proof of \cref{thm:MainTheorem2}]
For $\tau =1$, the contour-integral formula  \eqref{eqn:PsiHComp1} simplifies to
\begin{multline}
  \label{eqn:PsiHTauOne}
  (\psi_N)_{a_1,\dots,a_n} = \oint \cdots \oint \prod_{k=1}^n \frac{\diff u_k}{2\pi \i}
  \prod_{k=1}^n \frac{(1+xu_k)(1+u_k+u_k^2)^{N-2n}}{u_k^{a_k}}\\
  \times\prod_{1\leqslant i \leqslant j \leqslant n}(1-u_i u_j)  \prod_{1\leqslant i < j \leqslant n} (u_j-u_i)(1+u_j+u_iu_j)(1+u_i+u_j).
\end{multline}
The residue theorem implies that this contour integral yields a polynomial in $x$ with integer coefficients of degree at most $n$. Moreover, if $a_i = i$ for each $i=1,\dots,m$, then the integrations with respect to $u_1,\dots,u_m$ are trivial. We find
\begin{multline}
  (\psi_N)_{1,\dots,m,a_{m+1},\dots,a_n}=\oint \cdots \oint \prod_{k=m+1}^n \frac{\diff u_k}{2\pi \i}
  \frac{(1+xu_k)(1+u_k)^{2m}(1+u_k+u_k^2)^{N-2n}}{u_k^{a_k-m}}\\
  \times\prod_{m+1\leqslant i \leqslant j \leqslant n}(1-u_i u_j)\prod_{m+1\leqslant i < j \leqslant n} (u_j-u_i)(1+u_j+u_iu_j)(1+u_i+u_j).
\end{multline}
The right-hand side is a polynomial in $x$ of degree is at most $n-m$.
\end{proof}

More generally, we note that the polynomiality in both $x$ and $\tau$ follows from \cref{prop:PsiHComp2} (and from the residue theorem). The coefficients are integers. 
We also observe that they are non-negative.

%%%%%%%%%%%%%%%%%%%%%%%%%%%%%%%%%%%%%%%%%%%%%%%%%%%%%%%%%%%%%%%%%%%%%%%%%%%%%%%%%%%%%%%%%%%%%
\subsubsection*{Parity}
%%%%%%%%%%%%%%%%%%%%%%%%%%%%%%%%%%%%%%%%%%%%%%%%%%%%%%%%%%%%%%%%%%%%%%%%%%%%%%%%%%%%%%%%%%%%%
\begin{proof}[Proof of \cref{thm:MainTheorem3}]
For $\tau=1$, we find by comparison of \eqref{eqn:PsiHComp1} and \eqref{eqn:PsiHComp2} the relation
\begin{equation}
 (\psi_N)_{N+1-a_n,\dots,N+1-a_1}(x)
   = x^n (\psi_N)_{a_1,\dots,a_n}(x^{-1}).
\end{equation}
It is equivalent to $\mathcal P|\psi_N(x)\rangle = x^n|\psi_N(x^{-1})\rangle$.
\end{proof}

%%%%%%%%%%%%%%%%%%%%%%%%%%%%%%%%%%%%%%%%%%%%%%%%%%%%%%%%%%%%%%%%%%%%%%%%%%%%%%%%%%%%%%%%%%%%%
\subsection{Scalar products}
\label{sec:ScalarProducts}
%%%%%%%%%%%%%%%%%%%%%%%%%%%%%%%%%%%%%%%%%%%%%%%%%%%%%%%%%%%%%%%%%%%%%%%%%%%%%%%%%%%%%%%%%%%%%
In this section, we prove \cref{thm:MainTheorem4}. The proof follows the lines of \cite{difrancesco:07_2,degier:09}, and uses an antisymmetriser identity. We recall that the antisymmetriser $\mathcal A f$ of a function $f$ of the variables $u_1,\dots,u_n$ is defined as
\begin{equation}
  (\mathcal Af)(u_1,\dots,u_n)=\sum_{\sigma} \text{sgn}\,\sigma\, f(u_{\sigma(1)},\dots, u_{\sigma(n)}).
\end{equation}
Here, the sum runs over all permutations $\sigma$ of $\{1,\dots,n\}$. We use two elementary properties of the antisymmetriser. First, the Vandermonde determinant can be written as an  antisymmetriser:
\begin{equation}
  \label{eqn:VandermondeAsym}
  \Delta(u_1,\dots,u_n)=\prod_{1\leqslant i < j \leqslant n} (u_j-u_i) = \mathcal A\left(\prod_{i=1}^n u_i^{i-1}\right).
\end{equation}
Second, if $f$ and $g$ are functions of $u_1,\dots,u_n$ then we have
\begin{equation}
  \label{eqn:IntegralAsym}
  \oint {\cdots} \oint \prod_{i=1}^n\frac{\diff u_i}{2\pi \i} (\mathcal A f)(u_1,\dots,u_n)\,g(u_1,\dots,u_n) = \oint \cdots \oint \prod_{i=1}^n\frac{\diff u_i}{2\pi \i}f(u_1,\dots,u_n)(\mathcal Ag)(u_1,\dots,u_n),
\end{equation}
where the integration contour of each $u_i$ is a positively-oriented curve around $0$, but no other singularity of the integrand.

\begin{proof}[Proof of \cref{thm:MainTheorem4}] We compute the overlap \eqref{eqn:DefF} for the vector \eqref{eqn:DefPsiH} with arbitrary $\tau$. In terms of the components, we obtain
\begin{equation}
 F_N
=\sum_{\epsilon_1,\dots,\epsilon_n=0,1} \alpha^{\sum_{i=1}^n \epsilon_i} (\psi_N)_{N-2(n-1) -\epsilon_1,N-2(n -2)-\epsilon_2,\dots,N-\epsilon_n}.
\end{equation}
We use the integral formulas \eqref{eqn:PsiHComp2} to rewrite this sum in terms of a contour integral:
\begin{multline}
  F_N = \oint \cdots \oint \prod_{k=1}^n \frac{\diff u_k}{2\pi \i}\frac{(u_k+x)(u_k+\alpha)(1+\tau u_k + u_k^2)^{N-2n}}{u_k^{2k}}\\
  \times \prod_{1\leqslant i \leqslant j \leqslant n}(1-u_iu_j)\prod_{1\leqslant i<j\leqslant n} (u_j-u_i)(1+\tau u_j + u_i u_j)(\tau +u_i + u_j).
\end{multline}
The integrand contains a Vandermonde determinant. We use \eqref{eqn:VandermondeAsym} and rewrite it as an antisymmetriser. Using \eqref{eqn:IntegralAsym}, we obtain
\begin{equation}
  F_N = \oint \cdots \oint \prod_{k=1}^n\frac{\diff u_k}{2\pi \i\,u_k}\Biggl((u_k+x)(u_k+\alpha )(1+\tau u_k+u_k^2)^{N-2n}u^{k-1}_k
  (\tau +u_k)^{k-1}  \Biggr)g(u_1,\dots,u_n),
\end{equation}
where
\begin{equation}
  g(u_1,\dots,u_n)= \prod_{1\leqslant i \leqslant j \leqslant n}(1-u_iu_j)\mathcal A\left(\prod_{k=1}^n u_i^{-2k+1}\prod_{1\leqslant i<j\leqslant n} (1+\tau u_j + u_i u_j)\right).
\end{equation}

Let us denote by $g(u_1,\dots,u_n)_{\leqslant 0}$ the polynomial in $u_1^{-1},\dots,u_n^{-1}$ obtained from $g(u_1,\dots,u_n)$ by removing all monomials that contain at least one positive power in $u_1,\dots,u_n$. We have \cite{degier:09}
\begin{align}
  g(u_1,\dots,u_n)_{\leqslant 0} = \mathcal A\left(\prod_{i=1}^nu_i^{-i}(\tau+u_i^{-1})^{i-1}\right).
\end{align}
This identity allows us to apply \eqref{eqn:IntegralAsym} a second time. We obtain
\begin{multline}
  F_N = \oint \cdots \oint  \prod_{k=1}^n\frac{\diff u_k}{2\pi \i}(u_k+x)(u_k+\alpha)(1+\tau u_k+u_k^2)^{N-2n}u^{-k-1}_k(\tau+u_k^{-1})^{k-1}\\
  \times \Delta(u_1(\tau+u_1),\dots,u_n(\tau+u_n)).
\end{multline}
The Vandermonde determinant in the integrand allows us to rewrite $F_N$ as the determinant of a single contour integral:
\begin{align}
  F_N = \det_{i,j=1}^n\Biggl(\oint & \frac{\diff u}{2\pi \i}\,(u+x)(u+\alpha)(1+\tau u+u^2)^{N-2n}(\tau+u)^{i-1}(\tau+u^{-1})^{j-1}u^{i-j-2}\Biggr).
\end{align}
We evaluate the contour integral inside the determinant in terms of
\begin{multline}
  f_{i,j}^k=\oint \frac{\diff u}{2\pi \i}(\tau+u)^{i-1}(\tau+u^{-1})^{j-1} u^{i-j+k}\\ = \tau^{2(i-j)+k+1}\sum_{m=0}^\infty \binom{i-1}{2(i-j)+m+k+1}\binom{j-1}{m} \tau^{2m}.
\end{multline}
The sum on the right-hand side is finite and $f_{i,j}^k$ is, therefore, a polynomial in $\tau$. In terms of these polynomials, we find for even $N=2n$ the determinant
\begin{equation}
  \label{eqn:FEven}
  F_{2n} = \det_{i,j=1}^n\left(\alpha x  f_{i,j}^{-2}+(\alpha + x)f_{i,j}^{-1}+f_{i,j}^{0}\right).
\end{equation}
If $N=2n+1$ is odd then the evaluation of the contour integral, combined with the identity $f_{i,j}^k= f_{i,j+1}^{k+2}-\tau f_{i,j}^{k+1}$ and elementary column operations, yields
\begin{equation}
  \label{eqn:FOdd}
  F_{2n+1} = \det_{i,j=1}^n\left(\alpha x f_{i,j+1}^{0}+(\alpha + x)f_{i,j+1}^{1}+ f_{i,j+1}^{2}\right).
\end{equation}

Finally, we evaluate \eqref{eqn:FEven} and \eqref{eqn:FOdd} for $\tau =1$. Using
\begin{equation} 
   \left.f_{i,j}^k\right|_{\tau=1} = \binom{i+j-2}{2i-j+k},
\end{equation}
we obtain the expressions announced in \cref{sec:MainResults}.
\end{proof}

%%%%%%%%%%%%%%%%%%%%%%%%%%%%%%%%%%%%%%%%%%%%%%%%%%%%%%%%%%%%%%%%%%%%%%%%%%%%%%%%%%%%%%%%%%%%%
\subsection{The sum of the components}
\label{sec:TSASM}
%%%%%%%%%%%%%%%%%%%%%%%%%%%%%%%%%%%%%%%%%%%%%%%%%%%%%%%%%%%%%%%%%%%%%%%%%%%%%%%%%%%%%%%%%%%%%

In this section, we conjecture a relation between the sum of components
\begin{equation}
 S_N
  = \sum_{1\leqslant a_1 < \dots < a_n \leqslant N} (\psi_N)_{a_1,\dots,a_n}
\end{equation}
and a weighted enumeration of the so-called totally-symmetric alternating sign matrices (TSASMs). To this end, we recall that a square matrix $A = (a_{ij})_{i,j=1}^M$ of size $M$ is called an alternating sign matrix if \textit{(i)} the entries $a_{ij}$ take the values $-1,0,1$, \textit{(ii)} each row and column sum of $A$ is equal to one, and \textit{(iii)} along each row and column of $A$ the non-zero entries alternate in sign. We say that $A$ is a TSASM if it is invariant under all symmetries of the square. \Cref{fig:TSASMExample} shows an example of a TSASM.
\begin{figure}
\centering
\begin{tikzpicture}
\draw (0,0) node {$
  \left(
  \begin{array}{ccccccccccc}
 0 & 0 & 0 & 0 & + & 0 & 0 & 0 & 0 \\
 0 & 0 & + & 0 & - & 0 & + & 0 & 0 \\
 0 & + & - & 0 & + & 0 & - & + & 0 \\
 0 & 0 & 0 & + & - & + & 0 & 0 & 0 \\
 + & - & + & - & + & - & + & - & + \\
 0 & 0 & 0 & + & - & + & 0 & 0 & 0 \\
 0 & + & - & 0 & + & 0 & - & + & 0 \\
 0 & 0 & + & 0 & - & 0 & + & 0 & 0 \\
 0 & 0 & 0 & 0 & + & 0 & 0 & 0 & 0 
  \end{array}\right)
  $};
 \end{tikzpicture}
\caption{A totally-symmetric alternating sign matrix of size $M=9$, where $\pm$ represents the non-zero entry $\pm 1$.}
\label{fig:TSASMExample}
\end{figure}

The invariances of a TSASM $A$ of size $M$ are equivalent to the relations
\begin{equation}
  a_{ij} = a_{ji} = a_{i(M+1-j)},
  \label{eqn:TSASMInvariances}
\end{equation}
for each $i,j=1,\dots,M$. These relations have no solution for even size $M=2m$. Therefore, we focus on the case of odd size $M=2m+1$, where they have at least one solution \cite{bousquet:95,robbins:00}. Each solution has the property that the horizontal and vertical median of the matrix are fixed to alternating sequences of $+1$ and $-1$:
 We have
\begin{equation}
 a_{m+1\,j} = (-1)^{j+1}, \quad a_{i\,m+1} = (-1)^{i+1},
\end{equation}
for each $i,j=1,\dots,2m+1$. The medians divide the matrix into four sub-matrices of size $m$. The relations \eqref{eqn:TSASMInvariances} imply that all of their entries can be obtained from the entries of the following triangular part of the upper-left sub-matrix
\begin{equation}
  \begin{array}{cccc}
    a_{11} & a_{12} & \cdots & a_{1m}\\
     & a_{22} & & a_{2m}\\
    & & \ddots & \vdots \\
    & & & a_{mm}
  \end{array}.
  \label{eqn:TriangularPart}
\end{equation}
\Cref{fig:TSASMFundamentalDomains} displays the triangular parts of all TSASMs of size $M=9$.

To introduce a weighted enumeration, we define for each TSASM $A$ of size $M=2m+1$ the numbers $\mu(A)$ and $\nu(A)$ of its non-zero entries along and above the diagonal of the triangular part \eqref{eqn:TriangularPart}, respectively. In terms of the entries, they are given by
\begin{equation}
  \mu(A)=\sum_{i=1}^m |a_{ii}|, \quad \nu(A) = \sum_{1\leqslant i < j \leqslant m} |a_{ij}|.
\end{equation}
We use them to assign to $A$ the weight $t^{\mu(A)}\tau^{\nu(A)}$. By means of the weights, we introduce the following generating function:
\begin{equation}
A_{\mathrm{TS}}(2m+1;t,\tau) = \sum_{A}t^{\mu(A)}\tau^{\nu(A)}.
\end{equation}
Here, the sum runs over the set of all TSASMs $A$ of size $M=2m+1$.

For certain values of $t$ and $\tau$, this generating function yields integer sequences that enumerate symmetry classes of alternating sign matrices. For $t=0$ and $\tau=1$, one finds the numbers $A_{\mathrm{VOS}}(2m+1) = A_{\mathrm{TS}}(2m+1;0,1)$ of vertically and off-diagonally symmetric alternating sign matrices \cite{okada:04,fischer:19}. These numbers can, for example, be computed from Pfaffian formulas. For $t=\tau=1$, one obtains the TSASM numbers $A_{\mathrm{TS}}(2m+1) = A_{\mathrm{TS}}(2m+1;1,1)$. For $m=0,\dots,9$ they are given by
\begin{equation}
  A_{\mathrm{TS}}(2m+1)=1, 1, 1, 2, 4, 13, 46, 248,1516, 13654,\dots
\end{equation}
To our best knowledge, no explicit formula for these numbers is currently known.

\begin{figure}
\centering
\begin{tikzpicture}
  \draw (-4.5,0) node
  {$\begin{array}{cccc}
    0 & 0 & 0 & 0\\
      & 0 & 0 & +\\
      &   & 0 & 0\\
      &   &   & 0
  \end{array}$};
  \draw (-1.5,0) node
  {$\begin{array}{cccc}
    0 & 0 & 0 & 0\\
      & + & 0 & 0\\
      &   & 0 & 0\\
      &   &   & +
  \end{array}$};
  \draw (1.5,0) node
  {$\begin{array}{cccc}
    0 & 0 & 0 & 0\\
      & 0 & + & 0\\
      &   & - & 0\\
      &   &   & +
  \end{array}$};
  \draw (4.5,0) node
  {$\begin{array}{cccc}
    0 & 0 & 0 & 0\\
      & 0 & 0 & +\\
      &   & + & -\\
      &   &   & +
  \end{array}$};
  
  \draw(-4.5,-1.25) node {$\tau$};
  \draw(-1.5,-1.25) node {$t^2$};
  \draw(1.5,-1.25) node {$t^2\tau$};
  \draw(4.5,-1.25) node {$t^2\tau^2$};
  
\end{tikzpicture}
\caption{The upper-triangular parts of the upper-left sub-matrix for each of the four totally-symmetric alternating sign matrices of size $M=9$ and their weights.}
\label{fig:TSASMFundamentalDomains}
\end{figure}

We generated all TSASMs for $m=0,\dots,9$ with \textsc{Mathematica}, by exploiting a bijection between the triangular parts considered above and the configurations of a six-vertex model on a triangular domain \cite{lienardy:20}. Moreover, we computed the corresponding generating functions $A_{\mathrm{TS}}(2m+1;t,\tau)$. (For $m=0,\dots,7$, they are listed in \cref{tab:TSASMPolys}.)
\begin{table}[ht]
\centering
\begin{tabular}{c|l}
  $m$ & $A_{\mathrm{TS}}(2m+1;t,\tau)$\\
  \hline
  $0$ & $1$\\
  $1$ & $1$\\
  $2$ & $t$\\
  $3$ & $t (1+\tau)$\\
  $4$ & $\tau + t^2(1+\tau+\tau^2)$\\
  $5$ & $\tau(1+\tau^2)+t^2(1+3\tau+4\tau^2+2\tau^3+\tau^4)$\\
  $6$ & $t\tau(3+4\tau+8\tau^2+3\tau^3+2\tau^4)+t^3(1+3\tau+7\tau^2+6\tau^3+6\tau^4+2\tau^5+\tau^6)$\\
  $7$ & $t\tau(3+7\tau+17\tau^2+18\tau^3+15\tau^4+12\tau^5+4\tau^6+2\tau^7)$\\
  & $+t^3(1+6\tau+19\tau^2+32\tau^3+41\tau^4+35\tau^5+21\tau^6+11\tau^7+3\tau^8+\tau^9$
\end{tabular}
\caption{The generating functions $A_{\mathrm{TS}}(2m+1;t,\tau)$ for $m=0,\dots,7$.}
\label{tab:TSASMPolys}
\end{table}
By means of this computation, we observed that these polynomials are related the sums of the components $S_N=S_N(x,\tau)$. This observation prompts us to formulate the following conjecture:
\begin{conjecture}
\label{conj:TSASM}
 We have
   \begin{equation}
     S_N(x,\tau) = (1+x(x-\tau))^{n/2}A_{\mathrm{TS}}(2N+1;t,\tau),
   \end{equation}
  where $t=(1+x)/(1+x(x-\tau))^{1/2}$.
\end{conjecture}

This conjecture has several interesting consequences. First, it implies that the generating functions for the weighted enumeration of TSASMs introduced in this section can be obtained from a contour-integral formula. Indeed, \cref{prop:PsiHComp2} implies the contour-integral formula
\begin{multline}
  \label{eqn:CISumOfComponents}
S_N(x,\tau)
= \oint\cdots \oint \prod_{k=1}^n\frac{\diff u_k}{2\pi \i}\frac{(u_k+x)(1+\tau u_k + u_k^2)^{N-2n}}{u_k^{N-n+k}\left(1-\prod_{j=1}^ku_j\right)}\prod_{1\leqslant i \leqslant j \leqslant N}(1-u_iu_j)\\
  \times \prod_{1\leqslant i < j \leqslant n}(u_j-u_i)(\tau + u_i + u_j)(1+\tau u_j + u_i u_j).
\end{multline}
In particular, the TSASM numbers $A_{\mathrm{TS}}(2m+1)= S_{m}(0,1)$ can be obtained from this contour-integral formula (which appears to be the first one in the literature). Second, \cref{prop:X0XTau} implies the relation $S_N(0,\tau) = S_{N-1}(\tau,\tau)$. Hence, we obtain the curious equality
\begin{equation}
 A_{\mathrm{TS}}(2N+3;1,\tau)=A_{\mathrm{TS}}(2N+1;1+\tau,\tau).
\end{equation}
Finally, the conjecture implies that the sum of the components $S_N(x,1)$ of the spin-chain ground state discussed in \cref{sec:SpinChainGS} can be expressed in terms of the TSASM generating function. In particular, for the supersymmetric point it yields the TSASM numbers $S_N(1,1) = A_{\mathrm{TS}}(2N+3)$. We checked this relation up to $N=16$.

%%%%%%%%%%%%%%%%%%%%%%%%%%%%%%%%%%%%%%%%%%%%%%%%%%%%%%%%%%%%%%%%%%%%%%%%%%%%%%%%%%%%%%%%%%%%%
\section{Conclusion}
\label{sec:Conclusion}
%%%%%%%%%%%%%%%%%%%%%%%%%%%%%%%%%%%%%%%%%%%%%%%%%%%%%%%%%%%%%%%%%%%%%%%%%%%%%%%%%%%%%%%%%%%%%

In this article, we have introduced and studied a contour-integral solution of the boundary quantum Knizhnik-Zamolodchikov equations for the $R$-matrix and a diagonal $K$-matrix associated with the six-vertex model. In the homogeneous limit, it has allowed us to investigate the ground states of the open XXZ spin chain at the combinatorial point $\Delta=-1/2$ with special boundary magnetic fields. Our main results for the spin-chain ground states are \cref{thm:MainTheorem1,thm:MainTheorem2,thm:MainTheorem3,thm:MainTheorem4}. They provide the ground-state eigenvalue, several properties of their components and explicit formulas for certain scalar products and special components. Moreover, our investigation has led us to observe an intriguing relation between the homogeneous limit of the solution of the boundary quantum Knizhnik-Zamolodchikov equations and a generating function for a weighted enumeration of totally-symmetric alternating sign matrices.

Let us briefly discuss open problems and further directions. First, the present article contains two open conjectures. \Cref{conj:Overlaps} addresses scalar products involving several ground-state vectors. They allow one to compute the spin chain's bipartite fidelity and its multipartite generalisations. The conjecture's simplicity suggests not only that a proof could be found, but also that the fidelities' asymptotic expansions for large $N$ could be explicitly computed. A rigorous computation could confirm the predictions of conformal field theory for the leading terms of these expansions. \Cref{conj:TSASM} suggests a relation between the solution of the boundary quantum Knizhnik-Zamolodchikov equations and a weighted enumeration of TSASMs. The literature on alternating sign matrix enumeration and related integrable models offers, to our best knowledge, almost no results on TSASMs (let alone an explicit formula for the TSASM numbers). We hope that our conjecture will inspire a more profound investigation of TSASMs.
Second, it could be interesting to generalise the findings of the present article to higher spin. For periodic boundary conditions, such a generalisation has been found in \cite{fonseca:14}. It is based the construction of polynomial solutions to higher-spin quantum Knizhnik-Zamolodchikov equations from matrix elements of products of vertex operators. We hope to extend this construction to open boundary conditions in a future publication.

%%%%%%%%%%%%%%%%%%%%%%%%%%%%%%%%%%%%%%%%%%%%%%%%%%%%%%%%%%%%%%%%%%%%%%%%%%%%%%%%%%%%%%%%%%%%%%
\subsubsection*{Acknowledgements}
%%%%%%%%%%%%%%%%%%%%%%%%%%%%%%%%%%%%%%%%%%%%%%%%%%%%%%%%%%%%%%%%%%%%%%%%%%%%%%%%%%%%%%%%%%%%%%

This work was supported by the Fonds de la Recherche Scientifique-FNRS and the Fonds Wetenschappelijk Onderzoek-Vlaanderen (FWO) through the Belgian Excellence of Science (EOS) project no. 30889451 ``PRIMA – Partners in Research on Integrable Models and Applications''. 

%%%%%%%%%%%%%%%%%%%%%%%%%%%%%%%%%%%%%%%%%%%%%%%%%%%%%%%%%%%%%%%%%%%%%%%%%%%%%%%%%%%%%%%%%%%%%
% Bibliography
%%%%%%%%%%%%%%%%%%%%%%%%%%%%%%%%%%%%%%%%%%%%%%%%%%%%%%%%%%%%%%%%%%%%%%%%%%%%%%%%%%%%%%%%%%%%%

\end{document}